\documentclass[12pt,a4paper]{article}
\pdfoutput=1
\usepackage[utf8]{inputenc}
\usepackage{fullpage}
\usepackage{hyperref}
\usepackage{amssymb}
\usepackage{amsthm}
\usepackage{graphicx}
\usepackage[outdir=./]{epstopdf}
\usepackage[export]{adjustbox}
\usepackage{color}
\usepackage{enumitem}
\usepackage{empheq}
\usepackage{cite}
\usepackage{slashed}
\usepackage{amsmath}
\usepackage{comment}
\usepackage{dsfont}
\usepackage{braket}
\usepackage{simpler-wick}
\usepackage{simplewick}
\usepackage{tikz,tikz-feynman}
\numberwithin{equation}{section}
\usepackage{pgfplots}
\ifdim\overfullrule>0pt 
\usepackage{environ}

\NewEnviron{tikzpicture}{%
	\begin{pgfpicture}
		\pgfpathrectanglecorners{\pgfpointorigin}{\pgfpoint{3cm}{3cm}}%
		\pgfusepath{stroke}\end{pgfpicture}%
}
\fi

\usepackage{cleveref}
\crefname{section}{Section}{Sections}
\crefname{appendix}{Appendix}{Appendices}

\newtheorem{theorem}{Theorem}[section]

\newtheorem*{principle*}{Hamilton's principle}
\theoremstyle{definition}

\newtheorem{defi/}{Definition}[section]

\begin{document}

\titlepage

\begin{flushright}
	MS-TP-24-34\\
TUM-HEP-1547/24\\
CERN-TH-2022-143
\end{flushright}

\vspace*{1.2cm}

\begin{center}
	{\Large \bf  Generalized Wilson lines
		\\	\vspace*{0.5cm}
	 	   and the gravitational scattering 
	 	   of spinning bodies
	 	 
	}
	
	\vspace*{1.5cm} \textsc {D. Bonocore$^{ \sharp}$, A. Kulesza$^{ \flat}$, J. 
	Pirsch$^{ \flat}$} \\
	
	\vspace*{1cm} 
	
$^\sharp$ {Physik Department T31,  Technische 
	Universit\"at M\"unchen, \\
	James-Franck-Straße 1, D-85748, 
	Garching, Germany}	
	
\vspace*{0.2cm} 
	
$^\flat$ Institut f\"{u}r Theoretische Physik, Universit\"{a}t M\"{u}nster, Wilhelm-Klemm-Stra\ss e 9,
	D-48149 M\"{u}nster, Germany\\

\end{center}

\vspace*{7mm}

\begin{abstract}
	\noindent 
A generalization of Wilson line operators at subleading power in the soft 
expansion has been recently introduced as 
an 
efficient building block of gravitational scattering 
amplitudes for non-spinning objects. The classical limit in this picture 
corresponds to the strict Regge limit, where the Post-Minkowskian 
(PM) expansion corresponds to the soft expansion, interpreted as a sum over 
correlations of soft emissions.
Building on the well-studied worldline model with ${\cal N}=1$ supersymmetry, 
in this work we extend the generalized Wilson line (GWL) approach to the case 
of spinning gravitating 
bodies. 
Specifically, at the quantum level we derive from first-principles a 
representation for the spin $1/2$ GWL
that is relevant for the all-order factorization of next-to-soft gravitons with 
fermionic matter, thus generalizing the exponentiation of single-emission 
next-to-soft theorems.  
At the classical level, we identity the suitable generalization of Wilson line 
operators that enables the generation of classical spin observables at linear order in spin. 
Thanks to the crucial role played by the soft expansion, the map from Grassmann variables to classical spin is manifest. We also comment on 
the relation between the GWL approach and the  Worldline Quantum Field Theory as well as the Heavy Mass Effective Theory formalism. 
 We validate the approach  
by rederiving known results in the conservative sector at 2PM order.

\end{abstract}


\newpage

\section{Introduction}
\label{sec:intro}

Wilson lines are ubiquitous objects in quantum field theories. In particular, when evaluated 
on 
 paths composed of straight (infinite) segments, 
they exhibit universal renormalization properties \cite{Korchemsky:1993uz} that 
are extremely powerful in a variety of perturbative calculations. For example, 
 they provide an invaluable tool for a gauge 
invariant formulation of factorization theorems in QCD 
\cite{Collins:1981uk} which are vital for an all-order resummation of many 
collider 
observables. They are crucial building blocks for the construction
of effective field theories both in gauge theories \cite{Bauer:2000yr} and 
gravity \cite{Beneke:2012xa} and the systematic
treatment of power corrections \cite{Beneke:2002ph, Larkoski:2014bxa, 
Beneke:2022ehj}.  
They are also of primary importance for supersymmetric 
theories, where maximally helicity violating (MHV) $n$-legs scattering 
amplitudes in the planar limit are dual to 
$n$-cusped Wilson loops \cite{Drummond:2008vq}. However, a naive application of 
Wilson lines is sometimes not enough to cope 
with more general problems. In these situations, in order to minimally modify 
the formalism, it is often
 desirable to generalize these Wilson line operators. This is the case for the 
 aforementioned duality, 
where supersymmetric generalizations have 
shown how to extend the relation beyond the MHV case \cite{Caron-Huot:2010ryg}. 
It is also the case for 
the soft expansion, where Generalized 
Wilson Lines (GWLs) have been defined in gauge theories  
\cite{Laenen:2008gt, Bonocore:2020xuj} and gravity \cite{White:2011yy, 
Bonocore:2021qxh} in 
order to include power corrections in factorization theorems, thus reorganizing 
the effect of next-to-soft theorems to an arbitrary 
number of emissions. 

In the light of this wide range of applications, it is natural to ask whether 
(generalized) Wilson lines could provide a useful representation for the 
development of analytic tools in the gravitational wave program. In fact, a 
wealth of new methods has been proposed over the recent years, some of them 
driven by the possibility to compute observables in classical General 
Relativity (GR) 
with Quantum Field Theory (QFT) techniques \cite{Kosower:2018adc, 
Kalin:2020mvi,Dlapa:2021npj,
	Mougiakakos:2021ckm,Riva:2021vnj,Mougiakakos:2022sic,Riva:2022fru,
	Mogull:2020sak,Jakobsen:2021zvh,Jakobsen:2023ndj,Cristofoli:2021vyo,Bastianelli:2021nbs,
	Comberiati:2022cpm, Comberiati:2022ldk,Mandal:2022ufb, 
	Bini:2020uiq,Bini:2019nra,Bini:2020wpo,Blumlein:2020pyo,
	Levi:2015msa,Levi:2020kvb,Kim:2022pou,Capatti:2024bid,Cristofoli:2021jas, Damgaard:2019lfh, Damgaard:2023vnx,Mougiakakos:2024nku}. 
In this regard, the possibility to map the bound state problem to 
scattering data 
\cite{Cheung:2018wkq,Kalin:2019rwq,Kalin:2019inp,Saketh:2021sri,Gonzo:2023goe,Cho:2021arx,Gonzo:2024zxo}
 makes the study of the classical limit of gravitational scattering amplitudes 
 particularly 
interesting. 
For weak 
	gravitational 
fields and large impact parameters, in particular, gravitating bodies can 
be represented by highly energetic particles glancing off each other and 
interacting via low energetic gravitons.     
Therefore, one could naively guess that
the soft expansion and the Regge limit, both of which have a natural 
representation in terms of Wilson lines \cite{Korchemskaya:1994qp, 
Korchemskaya:1996je,Nguyen:2023ibj}, might be key 
for an efficient identification of classical contributions.
 Motivated by this picture 
and building on more recent developments \cite{Melville:2013qca, Luna:2016idw}, 
in 
\cite{Bonocore:2021qxh} we argued that suitably defined {\it classical} 
Generalized Wilson lines provide indeed 
a natural language for the computation of observables in the Post Minkowskian 
(PM) expansion. 

More specifically, focusing on the non-spinning case in the conservative sector 
up to 2PM, in \cite{Bonocore:2021qxh} we showed how the 
classical limit of scattering amplitudes can be extracted by taking 
their strict Regge limit. This is in contrast to 
the 
traditional eikonal approach 
\cite{Amati:2007ak,Ciafaloni:2018uwe,DiVecchia:2021bdo,DiVecchia:2021ndb, Bellazzini:2022wzv} where 
the PM 
expansion corresponds to an expansion in\footnote{Throughout this work we 
denote 
	with $s,t,u$ the standard Mandelstam variables.}  $t/s$. To succeed in this 
	task, we 
	constructed the 
scalar GWL starting from the worldline quantization of the relativistic 
particle. This strategy, which is the starting point also for
the recent Worldline 
Quantum 
Field Theory (WQFT) approach \cite{Mogull:2020sak, 
Jakobsen:2021zvh, Jakobsen:2021lvp, Jakobsen:2021smu}, consists of generating   
 the perturbative series from a path  
integration over both the graviton field $h^{\mu\nu}$ and the trajectory $x$ of 
the 
gravitating body (neglecting of course quantum terms). 
This aspect differentiates both GWL and WQFT from worldline EFT approaches 
\cite{Goldberger:2004jt, 
	Kalin:2020fhe,Kalin:2022hph, Goldberger:2022ebt,Foffa:2019rdf,Foffa:2019hrb,Blumlein:2021txe,Blumlein:2021txj,Cheung:2023lnj,Cheung:2024byb}, where the only 
propagating field is the graviton and trajectory deflections are obtained  
by solving the 
equations of motion iteratively. 
However, unlike the WQFT approach, the GWL representation is derived by 
integrating 
over $x$ each external line of the scattering amplitude before any graviton 
integration, thus 
generating observables by a VEV of suitably defined operators. This provides a 
connection with the computation of 
generalized soft 
functions that have been a focus of recent studies in QCD phenomenology 	
\cite{Becher:2010pd, 
	Moult:2018jjd, Beneke:2018gvs, Bahjat-Abbas:2019fqa, Beneke:2024cpq}. 
As such, the GWL representation establishes a clear relation between the PM 
expansion and the 
soft expansion, interpreted as a sum over correlations among multiple soft 
emissions rather than a sum over the momentum of the graviton.\footnote{The 
	possibility to control the classical scattering of gravitating body via the 
	soft expansion has been also described in 
    	\cite{Guevara:2018wpp,DiVecchia:2021bdo,DiVecchia:2022piu,Alessio:2024onn}.}.

In this paper we generalize the work of 
\cite{Bonocore:2021qxh}
by considering the scattering of spinning objects, focusing separately on the quantum case with spin $1/2$ particles and contributions linear in spin at the classical level. 
Specifically, starting 
from the well-known ${\cal N}=1$ supersymmetric model for the relativistic 
particle in curved space, we first derive the exponentiation of soft gravitons 
dressing the external Dirac states of a quantum scattering amplitude in terms 
of {\it soft} GWLs. We then move to the classical problem and derive a suitable 
representation of 2PM observables in terms of a VEV of {\it classical} GWLs.
There are numerous reasons for carrying out this program.

At the pure classical level, the motivation is fairly obvious, 
given 
the high demand of accurate theoretical predictions in gravitational wave 
astronomy and the pivotal role of spin. In this regard, state-of-the-art 
  PM calculations include the 3PM scattering angle and impulse to quadratic 
  order in 
  spin 
 \cite{Liu:2021zxr,Damgaard:2023ttc,FebresCordero:2022jts},
 while at 2PM the observables are known to even higher orders in spin 
 \cite{Aoude:2023vdk,Aoude:2022thd,Bern:2022kto,Bautista:2021wfy,Bautista:2022wjf,Bohnenblust:2024hkw,Bjerrum-Bohr:2023iey}. In particular, in worldline supersymmetric models, higher-spin degrees of freedom are
  notoriously difficult to implement, because of well-known no-go theorems \cite{Howe:1988ft}, though
 a worldline formulation that includes spin effects at fourth-order in spin at 2PM has been recently proposed \cite{Haddad:2024ebn}.
We consider 
our work as a 
first step in the development of a new method which might facilitate the 
computation of higher spin corrections on the worldline \cite{newpaper}. 
Specifically, although the starting point is based on worldline supersymmetric 
models as in the WQFT \cite{Mogull:2020sak}, 
by integrating out all worldline 
degrees of 
freedom (i.e. trajectory $x$ and spin variables $\psi$) we obtain a suitable 
definition for classical GWLs such that its VEV generates classical 
observables, hence opening up the possibility to use renormalization techniques 
developed for Wilson line correlators. 

On top of this, at the classical level, we clarify   
an aspect which has not been addressed in 
previous supersymmetric 
worldline calculations and that becomes particularly transparent in our 
approach. 
Specifically, when defining classical spin on the worldline one has 
to map the worldline Grassmann variables $\psi$ representing quantum spin 
(which are a necessary ingredient of the supersymmetric 
theory) to the classical spin tensor $S^{\mu\nu}$ in GR, defined as the canonical momentum conjugated to the angular 
velocity of an extended body. The procedure discussed in WQFT
consists of mapping these supersymmetric degrees of 
freedom to the worldline EFT, 
where spin is not represented by supersymmetric Grassmann variables, but rather 
by a tower of operators containing powers of $S^{\mu\nu}$ multiplied by the 
corresponding matching coefficients.  
In our approach, however, the soft expansion in the 
construction of the GWL removes the need to map to the EFT. Specifically, 
as we will discuss, the presence of the Lorenz spin generator $\sigma_{\mu\nu}$ 
in the soft GWL 
yields a map between the supersymmetric 
Grassmann variables and the classical 
spin tensor, such that the classical GWL can be derived in terms of a classical 
spin tensor. 

However, the motivation for this work is not restricted to the pure classical regime. 
Indeed, at the quantum level, an all-order treatment of subleading soft
 theorems for spinning particles \cite{Cachazo:2014fwa} 
has continued to attract  
attention in the recent years \cite{Guevara:2018wpp, Beneke:2022pue}. It is therefore natural to 
ask whether the soft GWL representation of 
\cite{White:2011yy} 
and \cite{Bonocore:2021qxh} can be extended to the case of fermionic matter. 
This
generalization is not immediate. Indeed, while it is 
clear that for a fixed number of emissions one should recover the result of 
soft theorems, thanks to which spin effects are described by a 
coupling to the Lorentz spin generator, the role of this
generator to all orders in perturbation theory is not immediately evident in a diagrammatic analysis.  
In 
particular, it is not clear how correlated subleading soft emissions depend on the spin of the 
hard emitter. Secondly, 
in the gauge theory case it has been shown in \cite{Bonocore:2020xuj} that
worldline supersymmetry is key to prove that the background field in the 
numerator of dressed propagators does not contribute in the asymptotic limit, 
making the derivation similar to the more studied quantization on a closed 
worldline which 
yields the one-loop effective action. Does a similar mechanism 
occur in gravity? 
We shall see that by deriving the soft GWL from 
first-principles in the 
supersymmetric 
worldline model we will 
generalize single-emission on-shell soft theorems to the case of an arbitrary number of 
off-shell graviton emissions, including pairwise correlations, and we will thus
address all these questions.

The parallel between the soft and the classical GWL is not only useful for 
the identification of the map to classical spin but also to 
highlight the 
different exponentiation of classical and quantum spin variables at a given 
order in the 
soft expansion. Specifically, as we shall discuss, the enhancement of spin in 
the classical limit 
yields a dependence over classical spin both in the leading eikonal (E) and in 
the subleading (i.e. next-to-eikonal (NE)) terms, in sharp contrast with the 
soft GWL which depends on the quantum spin only through single emissions at NE 
level. 

Finally, we shall note that constructing the 
spinning GWL from first principles in the supersymmetric worldline model enables a  
relation between the WQFT and other amplitude-based 
methods.
In particular, the relation with the Heavy Mass EFT (HEFT) of 
\cite{Brandhuber:2021eyq,Brandhuber:2023hhy,Brandhuber:2023hhl} emerges neatly, although classical observables in terms of on-shell amplitudes in a 
heavy mass expansion seems quite different from a supersymmetric worldline 
model where the classical limit is achieved at the Lagrangian level. The GWL 
approach is somewhat intermediate and can thus be used to provide a map between 
the two approaches, and more generally between worldline models and on-shell 
methods. 

The structure of the paper is as follows. In \cref{sec_worldlineexp}
we derive the factorization properties of subleading soft gravitons dressing 
the external states of a quantum scattering amplitude in terms of soft GWLs.
In doing so, we devote special care to the quantization of the supersymmetric 
model since, in order to derive the GWL, we adopt conventions which are 
somewhat less common in the worldline literature, such as $px$-ordering and a 
worldline that extends from a localized hard interaction to infinity.  
In \cref{sec_ampltidues_and_classical} we consider the classical limit, 
and 
derive a representation for the classical GWL. We cross-check the formalism 
with known results at 2PM order. We then comment on the relation between the 
WQFT and the HEFT approaches. We finally conclude in 
\cref{sec_concl}. Some technical calculations are presented in separated 
appendices. 
\section{Soft exponentiation on the 
spinning worldline}\label{sec_worldlineexp}
Building on the well-known $\mathcal{N}=1$ supersymmetric 
model for the relativistic particle in curved space \cite{Howe:1988ft}, in this 
section we discuss the exponentiation of an arbitrary number of (subleading) 
soft 
graviton emissions dressing a fermionic external state of a quantum scattering 
amplitude. Specifically, we provide a first-principles derivation of the 
fermionic 
gravitational generalized 
Wilson line, thus generalizing the scalar case discussed in 
\cite{White:2011yy} and \cite{Bonocore:2021qxh}.

A number of features make the spinning case distinct from the scalar case 
and related work discussed in the literature \cite{Bastianelli:2002qw,Bastianelli:2006rx,Bastianelli:2013tsa,Bastianelli:2023oyz}. First of all, we consider 
dressed 
propagators extending from a localized hard interaction to infinity. This 
choice, which  
was also made in previous work on the GWL \cite{White:2011yy, Laenen:2008gt, 
	Bonocore:2020xuj, Bonocore:2021qxh}, is particular to the case of soft 
	exponentiation and will be revisited in the Regge limit discussed in the 
	second part of this paper. It is a somewhat less explored case in the 
	literature 
	about the worldline formalism, where the length of the line is usually 
	assumed to be finite or, as in the case of the WQFT, from $-\infty$ to 
	$+\infty$.
On a more technical side, in contrast to the existing 
literature 
on the quantization of the ${\cal N}=1$ model that assumes Weyl ordering, we 
choose a rather 
uncommon $px$-ordering prescription\footnote{In
	this work $px$-ordering means that 
	we put all $\hat p$'s to the left of the $\hat x$'s 
	(e.g. $[\hat x\hat p]_{\text{px}}=\hat p\hat x+[\hat x,\hat p]$). For a 
	comparison with the more standard Weyl-ordering see 
	\cite{Bonocore:2021qxh}.}. This choice is however purely 
conventional, motivated by the fact that a path integral representation for 
the GWL in time-slicing regularization is more easily derived from a dressed propagator
between an initial state of definite position and a final state of definite momentum. 

We begin our discussion in 
\cref{sec:susy} with a brief review of how a dressed propagator for a spin $1/2$ particle in curved space
can be expressed in terms of the supersymmetric charges constructed from the corresponding worldline degrees of freedom.
We then work out a path integral representation for this dressed propagator in \cref{sec:setup}. In doing so, we pay special attention 
to some technical details 
concerning the construction of the fermionic path integral. In particular, 
following a well-established procedure that goes by the name of fermion 
doubling \cite{Bastianelli:2006rx}, we add 
spurious 
fermionic degrees of freedom (that we eventually integrate out) to cope with 
the real nature of the Grassmann fields. In this construction, we also 
show the elegant role that worldline 
supersymmetry plays in demonstrating that
the background soft graviton field does not contribute to the numerator of the 
dressed propagator in the asymptotic limit. This feature, which is relevant for 
spinning emitters, was already discussed in the gauge theory case in 
\cite{Bonocore:2020xuj}. Equipped with the path integral representation, in \cref{sec:GWL}
we finally solve
the integral in a weak field expansion, obtaining the sought representation for the
soft Generalized Wilson line. 

\subsection{From supersymmetric charges to dressed propagators}
\label{sec:susy}
We begin by considering a free scalar point particle. Apart from convention
purposes, this is useful since the scalar contribution will be isolated when discussing the spinor case later on. In flat spacetime, the classical  
worldline 
action in phase space 
with canonical variables $x^\mu(t)$ and $p^\mu(t)$,  which fulfill the Poisson 
bracket 
$\{x^\mu,p_\nu\}_{\text{PB}}=-\delta^\mu_\nu$. Such action is obtained by 
enforcing the constraint\footnote{We adopt the $+,-,-,-$ signature throughout.}
\begin{equation}
-p_\mu p_\nu 
g^{\mu\nu}+m^2\equiv 2H=0~,
\end{equation}
with Lagrange multiplier $e(t)$, to ensure the 
on-shell condition. This constraint, 
which in Dirac terminology is 
first class and thus generates gauge transformations (time reparameterization), 
is conventionally defined 
to be (twice) the Hamiltonian $H$. The Klein-Gordon equation in flat spacetime  
 then trivially follows 
after canonical quantization, by identifying $\hat p=i\partial_{\mu}$ and 
requiring the physical states $|\phi_{\text{phys}}\rangle$ 
to fulfill $\hat H|\phi_{\text{phys}}\rangle=0$. 

The transition from flat to curved spacetime introduces some technical 
difficulties, mainly due to the metric $g_{\mu\nu}$ being a function of 
worldline operator  
$\hat x^\mu(t)$. In particular, in analogy with the gauge theory case discussed 
in \cite{Bonocore:2020xuj}, the definitions of the Hamiltonian and the 
corresponding
covariant momentum operator $\hat\Pi_\mu$ are 
sensitive to the choice of the operator ordering. 
For gravity in particular, as 
discussed in \cite{Bonocore:2021qxh}, one can define 
the following 
 Hamiltonian
\begin{align}
\hat H&
= 
-\frac{1}{2}\hat\Pi_\mu  g^{\mu\nu}(-g)^{1/2}\hat 
\Pi_\nu (-g)^{-1/2}+\frac{1}{2}m^2 \notag \\
&= \frac{- {g}^{\mu\nu}\hat{\Pi}_\mu\hat{\Pi}_\nu+m^2}{2}
+\frac{i}{2} {g}^{\mu\nu} {\Gamma}_{\mu\nu}^\varrho
\hat{\Pi}_\varrho~.\notag 
\end{align}
Although such Hamiltonian is hermitian, this definition requires the following 
(non-hermitian) momentum
\begin{align}
\hat \Pi_\mu&=(-g)^{-1/2}\hat p_\mu(-g)^{1/2}~.
\label{mom}
\end{align} 
In $px$-ordering
this yields the following Hamiltonian
\begin{align}
H_{px}
&=\frac{1}{2}\left(- {p}_\mu 
 {p}_\nu {g}^{\mu\nu}+m^2
+i {p}_\mu\left(\partial_\nu {g}^{\mu\nu}
+ {g}^{\mu\nu} {\Gamma}^\alpha_{\nu\alpha}\right)\right)~.	
\label{Ham}
\end{align}
 
The inclusion of additional degrees of freedom such as color or spin can be 
achieved by enlarging the phase space by additional 
Grassmann fields. A particle of spin $\frac{\mathcal{N}}{2}$ in 
particular requires $\cal N$ real fermionic variables $\psi_i$ fulfilling 
$\{\psi_i^\mu,\psi_j^\nu\}_{\text{PB}}=-\eta^{\mu\nu}\delta_{ij}$. The corresponding action 
is constructed by 
adding 
$\mathcal{N}$ further constraints $Q_i\equiv p_{\mu}\psi_i^\mu=0$ and 
$J_{ij}\equiv i\psi^\mu_i\psi_{j\mu}$, with the corresponding Lagrange 
multipliers. These constraints generate the  
$\cal 
N$-extended supersymmetry algebra, which closes under Poisson brackets. 
In particular, one has
\begin{align}\label{algebra-N}
i\left\{Q_i,Q_j\right\}_{\text{PB}}=2H\delta_{ij}~.
\end{align} 
 After
quantization, the constraints lead  
to the massless Bargmann-Wigner 
equations for the multispinor wave function labeled by $\cal N$ spin indices. 
The massive case is more subtle, as it can be obtained by dimensional reduction 
of a 
massless model in $d+1$ dimensions. For practical calculations in four 
dimensions, it boils down to the introduction of an auxiliary 
variable $\psi^5$.

In this paper we focus on the massive spin $\frac{1}{2}$ case. In flat space, 
the corresponding worldline action 
takes the form
\begin{equation}\label{spin12Act}
	S=\int d\sigma \left(-p_\mu\dot{x}^\mu+\frac{i}{2}\psi^\mu\dot{\psi}^\nu\eta_{\mu\nu}-\frac{i}{2}\psi^5\dot{\psi}^5-eH-\chi Q\right).
\end{equation}
Here $\chi(t)$ and $e(t)$ denote the gravitino and the einbein fields, which 
together form a supergravity 
multiplet of this one-dimensional theory. 
The constraint $Q$ is then given by 
\begin{equation}\label{Qflat}
	Q=i(p_\mu \psi^\mu+m\psi^5).
\end{equation}
The Grassmann variables $\psi^\mu$ and $\psi^5$ fulfill the following Poisson bracket relations
\begin{align}
	\left\{\psi^\mu,\psi^\nu\right\}_{\text{PB}}&=-i\eta^{\mu\nu},\qquad
	\left\{\psi^5,\psi^5\right\}_{\text{PB}}=i,
	\label{poisson}
\end{align}
which immediately results in 
\begin{equation}\label{algebra-12}
	\left\{Q,Q\right\}_{\text{PB}}=i(p^2-m^2)=-2iH,
\end{equation}
in agreement with the general relation in \cref{algebra-N}.
Canonical quantization transforms \cref{poisson} to the following 
(anti-)commutation relations 
\begin{align}
	\left[\hat{x}^\mu,\hat{p}_\nu\right]&=-i\delta^\mu_\nu~,
	\qquad
	\left\{\hat{\psi}^\mu,\hat{\psi}^\nu\right\}=\eta^{\mu\nu}~,
	\qquad
	\left\{\hat{\psi}^5,\hat{\psi}^5\right\}=-1\label{eq:anticommutatorpsi5}~.
\end{align}
which are realized with 
\begin{align}
	\hat{p}_\mu&=i\partial_\mu,\qquad 
	\hat{\psi}^\mu=\frac{1}{\sqrt{2}}\gamma^\mu,\label{eq:operatorpsi}\qquad 
	\hat{\psi}^5=\frac{i}{\sqrt{2}}\gamma^5.
\end{align}
Correspondingly, $\hat{Q}$ then takes the form 
\begin{equation} 
	\hat{Q}=\frac{\gamma^5}{\sqrt{2}}(i\tilde{\gamma}^\mu\partial_\mu-m),
\end{equation}
where $i\gamma^5\gamma^\mu=\tilde{\gamma}^\mu$ is also a representation of the $\gamma$-matrices.
As expected, upon implementing the first class constraint for $\hat Q$ in Fock 
space 
via  $\hat{Q}\ket{\psi_{\text{phys.}}}=0$, we recover the standard Dirac 
equation in flat space. 
This observation is key in order to identify the constraints $\hat{Q}$ and 
$\hat H$ 
in curved 
space, as we are going to discuss. 

Classically, the covariant Dirac equation in curved space reads
\begin{equation}
	\left[i\gamma^ae_a^\mu\left(\partial_\mu-\frac{i}{2}\omega_\mu^{ab}\sigma_{ab}\right)-m\right]\Psi=0,  
\end{equation}
where $\omega_\mu^{ab}= e^a_\nu \partial_\mu e^{b\nu}+e^a_\nu\Gamma_{\mu\sigma}^{\nu} e^{b\sigma} $ is the spin connection expressed via the vierbein $e^a_\mu$ and $\sigma_{ab}=\frac{i}{4}\left[\gamma_a,\gamma_b\right]$ is the spin $1/2$ Lorentz generator. As usual in GR the $\gamma$-matrices are defined in the locally flat reference frame defined by the vierbein, which we denote with Roman indices.
Correspondingly, after making use of the momentum defined in 
\cref{mom} and 
recalling the flat space expression in \cref{Qflat}, it is not difficult to see 
that the
quantum mechanical operator $\hat{Q}$ leading to the Dirac equation in 
curved 
space should read
\begin{equation}\label{eq:Q}
	\hat{Q}=i\hat{\psi}^a\hat{e}_a^\mu
	\hat{\Pi}_\mu+im\hat{\psi}^5.
\end{equation}
Here, in analogy with \cref{mom}, we defined the covariant momentum as
\begin{equation}\label{momSpin}
\hat{\Pi}_\mu=\frac{1}{\sqrt{- {g}}}\hat{p}_\mu\sqrt{- {g}}
+\hat{\omega}_\mu^{ab}\hat{\sigma}_{ab}~,
\end{equation}
where
\begin{align}\label{sigma}
\hat{\sigma}_{ab}=\frac{i}{4}\left[\hat{\psi}_a,\hat{\psi}_b\right]~.
\end{align}	
The definition in \cref{momSpin} corresponds to a minimal coupling of the spinning particle to a 
gravitational background. 

Subsequently, to find $\hat{H}$, we exploit the fact that 
the supersymmetry algebra is closed, i.e.  
\begin{equation} 
	\left\{\hat{Q},\hat{Q}\right\}=2\hat{H}.
\end{equation}
After some tedious manipulations, one obtains  
\begin{equation}\label{eq:hamiltonianoperator}
	\hat{H}=\frac{-\hat{g}^{\mu\nu}\hat{\Pi}_\mu\hat{\Pi}_\nu+m^2}{2}
	+\frac{i}{2}\hat{g}^{\mu\nu}\hat{\Gamma}_{\mu\nu}^\varrho\hat{\Pi}_\varrho
	-\hat{R}_{ab}^{~~cd}\hat{\sigma}^{ab}\hat{\sigma}_{cd}
	~.
\end{equation}
This Hamiltonian bears some similarity to the scalar case of \cref{Ham}. The 
first and the second term in particular are of precisely the same form, with 
the covariant 
momentum of \cref{mom} modified as in \cref{momSpin} 
to include the spin dependent term. The last term instead contains the 
coupling of the spinning degrees of freedom to the curvature $R^{\mu\nu\rho\sigma}$, in 
analogy with the gauge theory case where one has the coupling between 
$\sigma_{\mu\nu}$ and the field strength tensor $F^{\mu\nu}$ 
\cite{Bonocore:2020xuj}.

After having identified the correct quantum operators $\hat Q$ and $\hat H$ 
that correspond to the 
classical first class constraints of the worldline model, we can finally 
consider the 
dressed propagator, which is defined as the matrix element of the ratio of these operators.
Specifically, we define 
 the dressed propagator $\mathcal{P}_1(p_f,x_i,\eta)$ for a spin $1/2$ 
particle in the background of a gravitational field as
\begin{equation}\label{dressed}
\mathcal{P}_1(p_f,x_i,\eta,\eta^5)
=\frac{1}{\braket{\bar{\chi}_f|\chi_i}\braket{\bar{\chi}^5_f|\chi^5_i}
	\braket{p_f|x_i}}
\bra{f}\frac{\hat{Q}}{-2\hat{H}+i\epsilon}\ket{i}~,
\end{equation}
where the initial and final states are
\begin{equation}\label{tensor}
\bra{f}=\bra{p_f}\otimes\bra{\bar{\chi}_f,\bar{\chi}^5_f}\text{~and~} 
\ket{i}=\ket{x_i}\otimes\ket{\chi_i,\chi_i^5}~.
\end{equation}
Here, fermions fulfill the twisted boundary conditions 
\begin{align}
\bar{\chi}_f+\chi_i=\sqrt{2}\eta~, \qquad 
\bar{\chi}_f^5+\chi_i^5=\sqrt{2}\eta^5~,
\label{bc}
\end{align}
where $\eta^\mu$ and $\eta^5$ are a 
set of constant Grassmann variables that eventually generate the spinor 
structure of the propagator. More specifically,  
as in the flat spacetime case where $g_{\mu\nu}=\eta_{\mu\nu}$, we need to 
map\footnote{In the literature this is called symbol 
	map\cite{Ahmadiniaz:2020wlm}.} 
\begin{align}\label{symbolmap}
&\eta^5 \mapsto \frac{i\gamma^5}{\sqrt{2}},\\
& \eta^{\mu_1}...\eta^{\mu_n}\mapsto\frac{1}{n! 
	2^{\frac{n}{2}}}\epsilon_{j_1,...,j_n}\gamma^{\mu_{j_1}}...\gamma^{\mu_{j_n}}~,
\end{align}
leaving us with standard gamma matrices.
In the next section we work out a path integral representation for the dressed 
propagator.

\subsection{Setting up the path integral}
\label{sec:setup}

A path integral representation for the dressed propagator
in \cref{dressed} can be obtained by using the curved space action corresponding to \cref{spin12Act} with the appropriate Hamiltonian 
\cref{eq:hamiltonianoperator} and integrating over all dynamical fields $x^\mu(t),p^\mu(t),\psi^\mu(t),e(t),\chi(t)$. In order to see the equivalence to the operator definition, it is customary to gauge fix the multiplet via $(e(t),\chi(t))=(T,\theta)$, so that the 
path 
integrations over $e(t)$ and $\chi(t)$ reduce to integrations over 
the 
proper time $T$ and the fermionic ``supertime'' $\theta$, respectively.
In this way we get
\begin{equation}\label{eq:dress}
\mathcal{P}_1(p_f,x_i,\eta,\eta^5)=e^{-ip_fx_i-\bar{\chi}_f\chi_i+\bar{\chi}^5_f\chi^5_i}\frac{1}{2}\int
d\theta\int\limits_{0}^{\infty}dT\, \mathcal{M}(T,\theta)~,
\end{equation}
where the integrations over the the remaining fields give the following matrix element
\begin{align}\label{eq:M}
\mathcal{M}(T,\theta)=
\bra{f}\exp\left\{-i\theta 
\hat{Q}-i(\hat{H}-i\epsilon)T \right\}\ket{i}~.
\end{align}
By comparing \cref{eq:M} with \cref{dressed}, we see that, in analogy with the gauge theory case of 
\cite{Bonocore:2020xuj}, the role of the $T$ is to exponentiate the numerator 
$\hat H$ 
while the 
Grassmann nature of $\theta$ allows the exponentiation of the denominator $\hat 
Q$.

We aim at a rigorous path integral representation for this matrix element. 
Thanks to \cref{tensor}, 
the representation for the bosonic part of \cref{eq:M} is the same as in the 
scalar case discussed in \cite{Bonocore:2021qxh}. In $px$-ordering and time slicing it reads 
\begin{align}\label{bosPath}
\bra{p_f}\exp\left\{i\hat{A}(\hat x,\hat p)T\right\}\ket{x_i}
=e^{ip_f\cdot 
x_i}\!\int\!\mathcal{D}x\mathcal{D}p 
\exp\left\{ip_fx(T)+i\!\int_{0}^{T}\!\!dt 
\left(-p\cdot\dot{x}+A_{\text{px}}(x,p)\right)\right\}~.
\end{align}
For most applications, $A_{\text{px}}(x,p)$ yields a trivial integration 
over $p$, and 
thus for 
perturbative calculations the only relevant propagator is the one for the $x$ 
field, 
which reads
\begin{align}\label{bosProp}
\contraction{}{x}{^a(t)}{x}x^\mu(t)x^\nu(t')
=-\eta^{\mu\nu}\min(t,t')~.
\end{align}
However,  it is sometimes necessary to consider derivatives of \cref{bosProp}, 
which at equal time contain $\delta(0)$ and are thus divergent.  
These divergences are canceled by auxiliary ghosts fields, which  
are introduced by exponentiating the factors $\sqrt{-g(x)}$ in the integration 
measure, which in turn emerge after integrating over $p$. The role of ghost 
fields 
in the construction of the GWL have been discussed in great detail in 
\cite{Bonocore:2021qxh} and will not be repeated here.
 We can safely ignore ghosts in the following. 

The construction of fermionic path integrals, on the other hand, relies on 
coherent states, which in turn require the notion of creation and annihilation 
operators. Since the supersymmetric $\mathcal{N}=1$ worldline model only 
contains a  single real (i.e. Majorana) Grassmann field $\hat\psi_a$, there is 
no 
natural way 
to 
construct creation 
and annihilation operators. A well-known possibility consists of doubling the 
fermionic 
degrees of freedom by introducing unphysical fields that eventually must be 
integrated out. Although this procedure has been thoroughly discussed   
in the literature \cite{Bastianelli:2006rx}, one has to carefully verify that 
the method works with the ordering prescription and the somewhat unusual 
boundary conditions we are 
adopting. We discuss this 
in the next section.

\subsubsection{Fermionic path integral}
We first rename
$\hat{\psi}^a\rightarrow \hat{\psi}^a_1$ and $\hat{\psi}^5\rightarrow 
\hat{\psi}^5_1$, and introduce free, unphysical fermionic operators 
$\hat{\psi}_2^a$ and $\hat{\psi}^5_2$. The linear combinations 
\begin{align}\label{fermdoubl}
\hat{\psi}_a&=\frac{1}{\sqrt{2}}\left(\hat{\psi}_1^a+i\hat{\psi}_2^a\right),
\qquad 
\hat{\psi}^5=\frac{1}{\sqrt{2}}\left(\hat{\psi}_1^5+i\hat{\psi}_2^5\right),
\end{align}
provide then the desired algebra, since they fulfill
\begin{align}
\left\{\hat{\psi}^a,\hat{\psi}^{b\dagger}\right\}&=\eta^{ab}\mathds{1},
\qquad 
\left\{\hat{\psi}^5,\hat{\psi}^{5\dagger}\right\}=-\mathds{1}.
\end{align}
This enables the construction of a Fock space, where the vacuum state 
$\ket{\Omega}$ is
annihilated by $\hat{\psi}^a$ and $\hat{\psi}^5$, while 
the excited states fulfill
\begin{align}
\ket{a}&\equiv\hat{\psi}^{\dagger}_a\ket{\Omega},\qquad 
\bra{a}\equiv\bra{\Omega}\hat{\psi}^a,\qquad 
\braket{a|b}=\delta^a_b~,
\end{align}
with analogous construction for $\psi^5$. A coherent state is then given by 
\begin{equation} 
\ket{\chi}=\exp\left\{\hat{\psi}^\dagger_a\chi^a\right\}\ket{\Omega},
\end{equation}
with normalization 
\begin{equation} 
\braket{\bar{\chi}|\chi}=e^{\bar{\chi}_a\chi^a}~.
\end{equation}

Equipped with coherent states, we now aim at writing a path integral 
representation for matrix 
elements of the form 
\begin{equation} 
\bra{\bar{\chi}_f}\exp\left\{\hat{A}(\hat{\psi},\hat{\psi}^\dagger)T\right\}\ket{\chi_i}~.
\end{equation}
We first write the identity operator in terms of coherent states as 
\begin{equation} 
\mathds{1}=\int d^4\bar{\chi}d^4\chi 
\ket{\chi}e^{-\bar{\chi}\chi}\bra{\bar{\chi}}~,
\end{equation}
with the product measure defined as
\begin{equation} 
d^n\bar{\chi}d^n\chi=\prod_{a=1}^{n}d\bar{\chi}_a d\chi^a~.
\end{equation}
As usual in time slicing, we split the length of the time interval $T$ into $N$ 
pieces $\tau$ and insert 
$N-1$ coherent state completeness relations $\mathds{1}_i$, defined as 
\begin{equation} 
\mathds{1}_i=\int d^4\bar{\chi}_id^4\chi_{i} 
\ket{\chi_i}e^{-\bar{\chi}_i\chi_{i}}\bra{\bar{\chi}_i}.
\end{equation}
Setting $\bar{\chi}_f=\bar{\chi}_N$ and 
$\chi_i=\chi_0$ we get
\begin{align}
\bra{\bar{\chi}_f}\exp\left\{\hat{A}(\hat{\psi},
\hat{\psi}^\dagger)T\right\}\ket{\chi_i}
=&\nonumber\bra{\bar{\chi}_N}e^{\hat{A}\tau}\mathds{1}_{N-1}\dots 
\mathds{1}_1e^{\hat{A}\tau}\ket{\chi_0}\\
=&\prod_{j=1}^{N-1}\left(d^4\bar{\chi}_jd^4\chi_{j}e^{-\bar{\chi}_j\chi_{j}}\right)\prod_{n=1}^{N}\bra{\bar{\chi}_n}e^{\tau\hat{A}}\ket{\chi_{n-1}}.
\end{align}
We then expand in $\tau$ the matrix elements between states at $n$ and $n-1$
\begin{equation} 
\bra{\bar{\chi}_n}\!e^{\hat{A}\tau}\!\ket{\chi_{n\!-\!1}}= 
\bra{\bar{\chi}_n}\!1+\hat{A}_{\text{px}}\tau+\dots\!\ket{\chi_{n\!-\!1}}= 
\exp\Big\{\! A_{\text{px}}(\chi_{n-1},\bar{\chi}_n) 
\tau+\mathcal{O}(\tau^2)\Big\}\braket{\bar{\chi}_n|\chi_{n\!-\!1}}.	
\end{equation}
Here we defined $px$-ordering for the Grassmann fields via \begin{equation} 
\left(\hat{\psi}^a\hat{\psi}^\dagger_b\right)_{\text{px}}
=\delta^a_b-\hat{\psi}^\dagger_b\hat{\psi}^a,
\end{equation}
in analogy to the bosonic operators and the function $A_{\text{px}}(\chi_{n-1},\bar{\chi}_n)$ is obtained by replacing in the ordered operator all operators with their eigenvalues on the coherent states. 
With the shorthand $A_n=A_{\text{px}}(\chi_{n-1},\bar{\chi}_n)$, we get
\begin{align}
&\bra{\bar{\chi}_f}\exp\left\{\hat{A}(\hat{\psi},\hat{\psi}^\dagger)T\right\}\ket{\chi_i}\nonumber\\
=&\prod_{j=1}^{N-1}\left(d^4\bar{\chi}_jd^4\chi_{j}\right)\exp\left\{\!\bar{\chi}_N\chi_{N-1}-\!\sum_{n=1}^{N-1}\bar{\chi}_n(\chi_{n}-\chi_{n-1})+\!\sum_{n=1}^{N}\tau
A_n\right\},
\end{align}
which in the continuum limit becomes
\begin{equation}\label{Mferm}
\bra{\bar{\chi}_f}\exp\left\{\hat{A}(\hat{\psi},\hat{\psi}^\dagger)T\right\}\ket{\chi_i}
=\int\limits_{\chi(0)=\chi_i}^{\bar{\chi}(T)=\bar{\chi}_f} 
\mathcal{D\bar{\chi}}\mathcal{D}\chi 
\exp\left\{\bar{\chi}_f\chi(T)+\int_{0}^{T}dt 
\left(-\bar{\chi}\dot{\chi}+A\right)\right\}.
\end{equation}
At this point it seems we achieved our goal, with \cref{Mferm} being the 
fermionic equivalent of \cref{bosPath}. However, the path integral in 
\cref{Mferm}
depends on the nonphysical imaginary part $\psi_2$. In principle this is not a 
problem, since by construction the final result must 
 depend only on the physical operator $\hat \psi_{1}$. The representation in 
 \cref{Mferm} is in fact what is typically used for computations in the literature 
 (see e.g. \cite{Bastianelli:2006rx}). However, in order to make manifest that 
 these 
 spurious degrees of 
 freedom do not alter physical predictions,  
 it would be desirable to
have a path integral representation that explicitly depends only on the 
physical fields. Since to the best of our knowledge this has not been discussed 
in the literature, and because we are adopting a rather unconventional ordering 
prescription, we present a first-principle construction for such 
representation   
in 
\cref{sec:ferm}.  
 The final result is that one can integrate out $\psi_2$ such that \cref{Mferm} 
 becomes 
\begin{equation}\label{fermionPath}
\bra{\bar{\chi}_f}\exp\left\{\hat{A}(\hat{\psi},\hat{\psi}^\dagger)T\right\}\ket{\chi_i}
=\int\limits_{\psi_1(0)+\psi_1(T)=2\eta}\!\!\!\!
 \mathcal{D}\psi_1 \exp\left\{\bar{\chi}_f\chi_i+\int_{0}^{T}dt 
\left(-\frac{1}{2}\psi_1\dot{\psi}_1+A(\psi_1)\right)\right\}~.
\end{equation} 
with the corresponding propagator
\begin{equation}\label{fermionProp}
\contraction{}{\psi}{_{1}^a(t)}{\psi}\psi_{1}^a(t)\psi_{1}^b(t')
=\frac{1}{2}\eta^{ab}(\theta(t-t')-\theta(t'-t))~.
\end{equation}
Few comments are in order here. First, the boundary condition of the path integral follows from \cref{bc} and the 
fact that  
$\psi_1(0)=\sqrt{2}\Re{(\chi_i)}$ and $\psi_1(T)=\sqrt{2}\Re{(\bar{\chi}_f)}$. 
Secondly, the continuum propagator for $\psi_{1}^a$ is 
in agreement with the intuitive picture of inverting the differential operator 
in the kinetic term. Finally, as mentioned, the operator $\hat{A}(\psi,\psi^\dagger)$ needs to be related to the function $A(\psi_1)$ via the chosen ordering prescription. We discuss this aspect in the next section.

\subsubsection{An ordering prescription for $\hat{Q}$ and $\hat{H}$}
 
The results of \cref{bosPath}
and  \cref{fermionPath}  yields a path integral 
representation for  
\cref{eq:M}.
However, there is still a crucial ingredient missing, namely that we need to 
map the operators $\hat Q$ of \cref{eq:Q} and $\hat H$ of 
\cref{eq:hamiltonianoperator} to classical expressions, in analogy to 
the scalar case of \cref{Ham}. In particular, on top of ordering the operators 
$\hat p$ and $\hat x$, one has to assign a classical phase space function to products of the operator $\hat\psi^a$, 
 contained both in $\hat Q$ and 
$\hat H$ through the operator $\hat{\sigma}^{ab}$ defined in \cref{sigma}. Once again, fermion doubling turns out to be 
a powerful tool.

First, using 
\cref{eq:anticommutatorpsi5} we $px$-order the bosonic operators in   
\cref{eq:hamiltonianoperator} thereby isolating the scalar contribution
\begin{align}\label{hamNew}
\hat{H}_{\text{px}}=\hat{H}_{\text{px}}^{\text{scalar}}-\hat{p}_\mu 
\hat{g}^{\mu\nu}\hat{\omega}_\nu+\frac{1}{2}\left[i\nabla_\nu\hat{\omega}^\nu-\hat{\omega}_\mu\hat{\omega}_\nu
 \hat{g}^{\mu\nu}\right]-\hat{R}_{ab}^{~~cd}\hat{\sigma}^{ab}\hat{\sigma}_{cd},
\end{align}
where $\hat{\omega}_\mu=\hat{\omega}_\mu^{ab}\hat{\sigma}_{ab}$ and the scalar 
Hamiltonian reads \cite{Bonocore:2021qxh}
\begin{align}
\hat{H}_{\text{px}}^{\text{scalar}}
=\frac{1}{2}\left(-\hat{p}_\mu 
\hat{p}_\nu\hat{g}^{\mu\nu}+m^2
+i\hat{p}_\mu\left(\partial_\nu\hat{g}^{\mu\nu}
+\hat{g}^{\mu\nu}\hat{\Gamma}^\alpha_{\nu\alpha}\right)\right).
\end{align} 
For the Grassmann operators, we replace $\hat\psi_1$ with $\hat\psi$ and $\hat\psi^\dagger$ by applying \cref{fermdoubl} and move all $\hat\psi^\dagger$'s to the left. After ordering, the operators can be evaluated on the coherent states at their time slice to obtain the phase space function $A_{\text{px}}(\bar\chi,\chi)$ depending on $\bar\chi$ and $\chi$. Integrating out the spurious imaginary part leaves just a dependence $A(\psi_1)$ on the original real Grassmann field $\psi_1$ (see \cref{fermionPath}). This lengthy procedure can be conveniently bypassed by means of the following Wick-like theorem for ordered products of Majorana 
operators, which reads
\begin{align}
\braket{\prod_{i=1}^{n}\hat{\psi}_1^{a_i}}=\sum\limits_{\begin{subarray}
	{c}\text{possible}\\\text{contractions}
	\end{subarray}}\left[\text{contractions}\times\prod\limits_{i\text{ not 
		contracted}}\braket{\hat{\psi}_1^{a_i}}\right]~.
\label{wick}
\end{align}
Here, we defined the expectation value via
\begin{equation}
\psi_1^a\equiv\braket{\hat{\psi}^a_1}=\braket{\frac{\hat{\psi}^a+\hat{\psi}^{a\dagger}}{\sqrt{2}}}\equiv\frac{1}{\braket{\bar{\chi}|\chi}}\bra{\bar{\chi}}\frac{\hat{\psi}^a+\hat{\psi}^{a\dagger}}{\sqrt{2}}\ket{\chi}=\frac{\bar{\chi}^a+\chi^{a}}{\sqrt{2}},
\end{equation}
and the contraction of two Majorana operators 
as 
\begin{equation}\label{wick-contr}
\contraction{}{\hat{\psi}}{_1^a }{\hat{\psi}}
\hat{\psi}_1^a  
\hat{\psi}_1^b=\frac{1}{2}\left\{\hat{\psi}_1^a,\hat{\psi}_1^b\right\}=\frac{\eta^{ab}}{2}~.
\end{equation}
We also adopted the convention
\begin{equation}
\contraction{}{\hat{\psi}}{^a_1 \hat{\psi}_1^b}{\hat{\psi}}\hat{\psi}^a_1 
\hat{\psi}_1^b 
\hat{\psi}_1^c=-\contraction{}{\hat{\psi}}{^a_1}{\hat{\psi}}\hat{\psi}^a_1 
\hat{\psi}_1^c \hat{\psi}_1^b~.
\end{equation}

The outcome of \cref{wick} is a correspondence between products of 
operators 
$\hat\psi_1$ and the function of $\psi_1$ to be used in the path integral. In particular, the operators appearing in the Hamiltonian and supersymmetry get replaced as follows
\begin{align} 
\hat{\sigma}_{ab}\rightarrow& 
\frac{i}{2}\left(\braket{\psi_a}\braket{\psi_b}+\wick{\c{\psi_a} 
	\c{\psi_b}}\right)=\frac{i}{2}\left(\psi_a\psi_b+\frac{\eta_{ab}}{2}\right)~,\\
\hat{\sigma}_{ab}\hat{\sigma}_{cd}\rightarrow&\nonumber 
-\frac{1}{4}\Big(\psi_a\psi_b\psi_c\psi_d+\frac{1}{4}(\eta_{ab}\eta_{cd}+\eta_{ad}\eta_{cb}-\eta_{ac}\eta_{bd})\\
&+\frac{1}{2}\left(\eta_{ab}\psi_c\psi_d+\eta_{cd}\psi_a\psi_b-\eta_{ac}\psi_b\psi_d-\eta_{bd}\psi_a\psi_c+\eta_{bc}\psi_a\psi_d+\eta_{ad}\psi_b\psi_c\right)\Big).
\end{align}
We further notice that when inserting these terms in the Hamiltonian  
their contributions drastically simplify due to the following contractions
\begin{align}\label{contr1}
\hat{\omega}_\mu^{ab}\hat{\sigma}_{ab}&= \frac{i}{2}\omega^{ab}_\mu 
\psi_a\psi_b,\\
\label{contr2}
\hat{R}^{abcd}\hat{\sigma}_{ab}\hat{\sigma}_{cd}&= 
-\frac{1}{4}R^{abcd}\psi_a\psi_b\psi_c\psi_d+\frac{1}{8}R=\frac{1}{8}R~,\\
\hat{g}^{\mu\nu}\hat{\omega}_\mu^{ab}\hat{\omega}_\nu^{cd}
\hat{\sigma}_{ab}\hat{\sigma}_{cd} &= 
-\frac{1}{4}g^{\mu\nu}\omega^{ab}_\mu\omega^{cd}_\nu 
\psi_a\psi_b\psi_c\psi_d+\frac{1}{8}
g^{\mu\nu}\omega_\mu^{ab}\omega_\nu^{ab}~, \label{contr3}
\end{align}
where in \cref{contr1} we made use of the antisymmetry of the spin connection 
and	 in \cref{contr2} we used $R^{a[bcd]}=0$. 

In conclusion, by making use of \cref{contr1}, \cref{contr2}
 and \cref{contr3} in 
\cref{hamNew} and \cref{eq:Q}, the classical expressions for $H$ and $Q$ that 
we use in the path integral 
take the following form 
\begin{align}\label{Qpx}
	Q_{\text{px}}=&i\left(\psi^a  e_a^\mu(p_\mu+\frac{i}{2}\omega_\mu^{cd}\psi_c\psi_d)+m\psi_5-\psi^a i\partial_\mu(e_a^\mu \sqrt{-g})\frac{1}{\sqrt{-g}}+\frac{i}{2}e^\mu_c \omega_\mu^{cd}\psi^d \right)
\end{align}
and
\begin{align}
H_{\text{px}}=&\nonumber
H_{\text{px}}^{\text{scalar}}-\frac{i}{2}p_\mu g^{\mu\nu}\omega^{ab}_\nu 
\psi_a\psi_b\\&-\frac{1}{4}\nabla_\nu 
g^{\nu\sigma}\omega_\sigma^{ab}\psi_a\psi_b+\frac{1}{8}g^{\mu\nu}
\omega^{ab}_\mu\omega^{cd}_\nu 
\psi_a\psi_b\psi_c\psi_d-\frac{1}{16}g^{\mu\nu}\omega_\mu^{ab}\omega_\nu^{ab}-\frac{1}{8}R~.
\label{Hpx}
\end{align}

With \cref{bosPath}, \cref{fermionPath}, \cref{Qpx} and \cref{Hpx}
we have all ingredients to write a path integral representation for the 
matrix element of \cref{eq:M}, which reads
\begin{align}\label{path}
	\mathcal{M}(T,\theta)=&\nonumber\smashoperator{\int\limits_{x(0)
			=x_i}^{p(T)=p_f}}~\mathcal{D}x~\mathcal{D}
		p~\smashoperator{\int\limits_{\psi_1(0)+\psi_1(T)=2\eta}}
		\mathcal{D}\psi_1~\smashoperator{\int\limits^{\psi^5_1(0)+\psi^5_1(T)=2\eta^5}}
		\mathcal{D}\psi^5_1	 
	\exp\Bigg\{ip_fx(T)+\bar{\chi}_f\chi_i
	-\bar{\chi}_f^5
	\chi^5_i\\&+i\int_{0}^{T}dt
	 \left(-p\dot{x}+\frac{i}{2}\psi_1\dot{\psi_1}
	-\frac{i}{2}\psi_1^5\dot{\psi}_1^5-H_{\text{px}}
	-\frac{\theta}{T}Q_{\text{px}}\right)\Bigg\}~,
\end{align} 
where the constraints $H_{\text{px}}$ and $Q_{\text{px}}$ refer to the 
px-ordered operators $\hat{H}_{\text{px}}$ and $\hat{Q}_{\text{px}}$ evaluated 
on the eigenstates of the phase space variables at time $t$.

\subsubsection{A trick with conserved charges and the asymptotic limit}

At this point, to perform the integration over $\psi^5$ and further simplify the expression in 
\cref{path},  we perform a trick that has been already exploited in the gauge 
theory case \cite{Bonocore:2020xuj}\footnote{See also 
\cite{Alexandrou:1998ia}.} and that becomes particularly useful when taking the 
asymptotic 
limit $T\to 
\infty$. We consider the integration over $\theta$ in \cref{eq:dress}
and \cref{eq:M} 
and we first simplify 
\begin{equation}\label{trick}
	\int d\theta\exp\left\{-i\theta\frac{1}{T}\int_{0}^{T} 
	Q_{\text{px}}(t)dt\right\}=-i\frac{1}{T}\int_{0}^{T} Q_{\text{px}}(t)dt.
\end{equation}
The integration over the phase space variables in \cref{path} thus corresponds 
to the expectation value $\braket{Q_{\text{px}}}$. However, $\hat{Q}$ is a 
conserved charge due to the supersymmetry of the worldline model and it is 
therefore constant. This can also be seen by observing that 
$\left[\hat{H},\hat{Q}\right]=0$. 
The above expression can therefore be 
simplified further into 
\begin{equation} 
	-i\frac{1}{T}\int_{0}^{T} Q_{\text{px}}(t)dt\rightarrow 
	-i\frac{1}{T}\int_{0}^{T} \braket{\hat Q_{\text{px}}(t)}dt=-i \braket{\hat 
	Q_{\text{px}}(t)}.
\end{equation}
Being a constant, $Q_{\text{px}}(t)$ can be evaluated at an arbitrary time. We 
choose $t=T$, to get
\begin{align}\label{pathNew}
	\mathcal{M}(T,\theta)=&\nonumber\braket{-i\theta \hat 
	Q_{\text{px}}(T)}\smashoperator{\int
	\limits_{x(0)=x_i}^{p(T)=p_f}}~\mathcal{D}x~\mathcal{D}
p~\smashoperator{\int\limits_{\psi(0)+\psi(T)=2\eta}}\mathcal{D}
\psi~\smashoperator{\int\limits^{\psi^5(0)+\psi^5(T)=2\eta^5}}\mathcal{D}\psi^5	
	\exp\Bigg\{ip_fx(T)\\&+\bar{\chi}_f\chi_i-\bar{\chi}_f^5\chi^5_i
	+i\int_{0}^{T}dt
	\left(-p\dot{x}+\frac{i}{2}\psi\dot{\psi}
	 	-\frac{i}{2}\psi^5\dot{\psi}^5-H_{\text{px}}\right)\Bigg\}.
\end{align}
where we dropped the subscript of $\psi_1$ for notational convenience.

 There 
are a number of advantages that follow from having factorized the expectation 
value of $\hat Q(T)$.
 The most obvious one is that the
integration over  $\psi^5$ can be done straightforwardly, since the 
dependence on $\psi^5$ is entirely in the free kinetic term, giving unity. Similarly, \cref{Hpx} depends quadratically on the momentum and thus 
the integral over $p^\mu$ is Gaussian and 
yields
\begin{align}\label{MT}
	\mathcal{M}(T,\theta)=&\nonumber\braket{-i\theta 
	Q_{\text{px}}(T)}
e^{ip_fx(T)+\bar{\chi}_f\chi_i-\bar{\chi}_f^5\chi_i^5-\frac{i}{2}m^2T}
	f(x^\mu,\psi,g_{\mu\nu},T)~.
\end{align}
where
\begin{align}
f(x^\mu,\psi,& g_{\mu\nu},T)=\nonumber\smashoperator{\int\limits_{x(0)=x_i}}\mathcal{D}
x~\smashoperator{\int\limits^{\psi(0)+\psi(T)=2\eta}}\mathcal{D}\psi 
\exp\Bigg\{i\int\limits_{0}^{T}dt \Bigg( 
\frac{i}{2}\psi\dot{\psi}-\frac{1}{2}\dot{x}^\mu\dot{x}^\nu 
g_{\mu\nu}+\frac{i}{2}\dot{x}^\mu 
g_{\mu\nu}g^{\tau\lambda}\Gamma^\nu_{\tau\lambda}+\\
&+\frac{i}{2}\dot x^\mu\omega_\mu^{ab}\psi_a
\psi_b+\frac{1}{4}g^{\mu\nu}\partial_\nu\omega_\mu^{ab}\psi_a\psi_b+\frac{1}{8}
\!\left(\!R+\frac{1}{2}g^{\mu\nu}\omega_\mu^{ab}\omega_\nu^{ab}
+g_{\mu\nu}g^{\alpha\beta}g^{\gamma\delta}\Gamma_{\alpha\beta}^\mu
\Gamma^\nu_{\gamma\delta}\right)\!\Bigg)\!\Bigg\}~.
\end{align}

However, to fully appreciate the factorization of $Q(T)$ above, we have to 
proceed as in the 
previous work on the GWL 
\cite{Laenen:2008gt, 
Bonocore:2020xuj, Bonocore:2021qxh} and LSZ truncate the dressed propagator.
In doing so, we parameterize the path integral over $x$ via
\begin{align}
x(t)=x_i+p_ft+\tilde{x}(t)~,
\label{xtilde}
\end{align}
 and eventually consider the on-shell limit $p_f^2\to m^2$. Specifically, 
 we consider the following chain of equalities
 \begin{align}
 &\bar 
 u(p_f)\frac{p^2_f\!-\!m^2\!+\!i\epsilon}{i(\slashed{p}_f\!+\!m)}	
 \mathcal{P}_1(p_f,x_i,\eta,\eta^5)\notag \\
 &=\bar 
 u(p_f)
 \int\limits_{0}^{\infty}dT\,\frac{d}{dT}\left(
 e^{\frac{i}{2}(p_f^2-m^2)T}\right) 	\int
 d\theta \theta \frac{	\braket{-iQ_{px}(T)}}{i(\slashed{p}_f\!+\!m)} 
 e^{-\frac{1}{2}\epsilon T}f(x_i,p_f,\eta,T) \notag \\
 &=\bar 
 u(p_f) \lim_{T\to \infty}
 \frac{	\braket{-iQ_{px}(T)}}{i(\slashed{p}_f\!+\!m)} 
 f(x_i,p_f,\eta,T)~,
 \end{align} 
 where we first integrated by parts, then we set the Feynman $\epsilon$ to zero 
 and only at the very end we took the on-shell limit. 
 In this way, we are left with the asymptotic boundary term i.e. the dressed 
 propagator  
 evaluated 
 at $T\to\infty$.

 The value of having the charge $\hat Q(t)$ evaluated at $t=T$ is now evident. 
 In fact, 
  what 
 seemed at first glance a very technical trick that allowed us to trivially 
 perform the integration over $\psi^5$, has in fact a clear
 physical 
 interpretation. By evaluating the charge $\hat Q$ asymptotically for 
 $T\to\infty$, we 
 are in fact 
 replacing the charge $Q$ in the presence of a gravitational background 
  defined in \cref{eq:Q} 
 with the 
 free charge in flat spacetime, i.e.
 \begin{align}
 \lim_{T\to\infty} \braket{\hat Q_{\text{px}}(T)}
 =\lim_{T\to\infty}(i\eta^a e_a^\mu(x(T))
  p_\mu +im\eta^5)=i\eta^\mu 
  p_\mu +im\eta^5,
 \end{align}
  Stated differently, the background  
 field does not contribute to the numerator $Q$ of the asymptotic propagator 
 (which 
 determines the structure of its spin indices) but only to the denominator $H$, 
 which 
 contains 
 the coupling between spin and curvature. We have thus extended the 
 gauge-theory result of \cite{Bonocore:2020xuj}, by showing that also in 
 gravitational theories there is a clear 
 connection between the worldline representation of an asymptotic dressed 
 propagator and
 the more studied worldline quantization on the circle (i.e. the one-loop 
 effective action), which by definition is constructed with denominators only 
 (see e.g. \cite{Bastianelli:2002fv}).

In conclusion, we are left with the following path-integral representation for 
the 
asymptotic 
propagator (i.e. the 
LSZ-truncated dressed propagator)
\begin{equation} 
	\frac{p^2\!-\!m^2\!+\!i\epsilon}{i(\slashed{p}\!+\!m)}\mathcal{P}_1(x_i,p_f,\eta)
	=\smashoperator{\int\limits_{x(0)=0}}\mathcal{D}x~\smashoperator{\int
		\limits^{\psi(0)+\psi(\infty)=2\eta}}\mathcal{D}\psi 
		\exp\Bigg\{i\int\limits_{0}^{\infty}dt\,
		e^{-\epsilon t}\, L_{\text{spin}}(t)\Bigg\}~,
\end{equation}
where
\begin{align}\label{Lexact}
	L_{\text{spin}}(t)=
	&L_{\text{scalar}}+\frac{i}{2}\psi\dot{\psi}+\frac{i}{2}(\dot{x}+p)^\mu
	\omega_\mu^{ab}\psi_a\psi_b+\frac{1}{4}g^{\mu\nu}\partial_\nu\omega_\mu^{ab}
	\psi_a\psi_b+\frac{1}{8}R+\frac{1}{16}g^{\mu\nu}\omega_{\mu}^{ab}\omega_{\nu}^{ab}~.
\end{align}

Here, $L_{\text{scalar}}$ is just the $x$ kinetic term together with the same counter terms in the second line as in \cite{Bonocore:2021qxh} (with eq. (4.13) inserted into eq. (4.19) in that paper) and reads
\begin{align}\label{Lscal}
L_{\text{scalar}}(t)=&\nonumber-\frac{1}{2}\dot{x}^\mu\dot{x}^\nu \eta_{\mu\nu}-\frac{1}{2}(p^{\mu}p^{\nu}+2p^\mu \dot{x}^\nu+\dot{x}^\mu\dot{x}^\nu)(g_{\mu\nu}-\eta_{\mu\nu})\\&+\frac{i}{2}(p+\dot{x})^\mu g_{\mu\nu}\Gamma^\nu_{\varrho\sigma}g^{\varrho \sigma}+\frac{1}{8}g_{\mu\nu}\Gamma^\nu_{\varrho\sigma}g^{\varrho \sigma}\Gamma^\mu_{\alpha\beta}g^{\alpha\beta}~.
\end{align} 
Note that although the last two terms in \cref{Lexact} are $\psi$-independent they are not included in $L_{\text{scalar}}$, since they are counterterms that originate from the ordering prescription for the Grassmann fields, and hence are absent in the scalar case.
We also dropped the notation from \cref{xtilde} and \cref{fermdoubl} and 
renamed the variables
$\tilde x \to x$, $p_f\to p$ and $\psi_{1}\to \psi$. Note that the dependence on $x_i$ is now hidden in the argument of the graviton field, which needs to be evaluated at $x_i+pt+x(t)$. The remaining integrations 
over $x$ and $\psi$ 
cannot be carried out exactly.
In the next section we will perform this task perturbatively in the soft 
expansion. The final result will be an exponential form for the LSZ-truncated 
dressed propagator that we call GWL.

\subsection{Generalized Wilson line}
\label{sec:GWL}

In the previous section we derived an exact result for the path integral 
representation of the LSZ-truncated Dirac 
propagator in the presence of a background gravitational field.
No assumption on this background has been made.
We now perform the first approximation and expand the metric around a 
Minkowski background by defining the graviton $h_{\mu\nu}$ via
\begin{align}
\label{graviton}
g_{\mu\nu}=\eta_{\mu\nu}+\kappa h_{\mu\nu}~,
\end{align}
where $\kappa^2=32\pi G$, with $G$ the Newton constant. 
Correspondingly, the Lagrangian of \cref{Lexact} expanded to second order in 
$\kappa$ reads 
\begin{align}
L_{\text{spin}}(t)=&L_{\text{scalar}}(t)
+\nonumber\frac{i}{2}\psi\dot{\psi}+\frac{i\kappa}{2}(\dot{x}^\mu+p^\mu)\partial^{b}h^{a}_\mu\psi_a\psi_b\\
&+\nonumber 
\frac{i\kappa^2}{2}(\dot{x}^\mu+p^\mu)\left(\frac{1}{2}\left(h^{a\lambda}\partial_\lambda
h_\mu^{~b}+h^{b\lambda}\partial^{a} 
h_{\mu\lambda}\right)\psi_a\psi_b+\frac{1}{4}h^{b}_{~\lambda}\partial_\mu 
h^{a\lambda}\psi_a\psi_b\right)\\
&+\nonumber \nonumber\frac{\kappa^2}{4}\Bigg(\partial^{b}\partial_\mu h^{\mu 
	a}-h^{\mu\nu}\partial_\nu 
\partial^{b}h^{a}_{~\mu}+\frac{1}{4}h^{b}_{~\mu}\square h^{\mu a}
\\&\nonumber+\frac{1}{2}\left(h^{a\lambda}\partial_\mu\partial_\lambda 
h_\mu^{~b}+\partial_\mu 
h^{b\varrho}\partial^{a}h_{\mu\varrho}+h^{b\varrho}\partial_\mu 
\partial^{a}h_{\mu\varrho}\right)\Bigg)\psi_a\psi_b+\\
&\nonumber+\frac{\kappa}{8}\left(\partial_\mu\partial_\nu h^{\mu\nu}-\square 
h\right)+\frac{\kappa^2}{8}\Big(h^{\mu\nu}\square 
h_{\mu\nu}+h^{\mu\nu}\partial_\mu\partial_\nu h-\left(\partial_\mu 
h^{\mu\alpha}-\frac{1}{2}\partial^\alpha h\right)^2\\
&+\frac{3}{4}\partial_\alpha h_{\mu\nu}\partial^\alpha 
h^{\mu\nu}-h^{\mu\nu}\partial_\mu\partial_\alpha 
h^\alpha_\nu-\frac{1}{2}\partial_\mu h_{\mu\alpha}\partial^\alpha 
h_{\mu\nu}+\frac{1}{2}\partial^{[\alpha}h^{\beta]\mu}\partial^{[\alpha}
h^{\beta]}_\mu\Big)~,
\label{Lspin}
\end{align}
where the square bracket in the last term denote index antisymmetrization and
\begin{align}
L_{\text{scalar}}(t)&=
\nonumber -\frac{1}{2}\dot{x}^\mu\dot{x}^\nu
\eta_{\mu\nu}-\frac{\kappa}{2}h_{\mu\nu}(\dot{x}^\mu\dot{x}^\nu+p^\mu 
p^\nu+2\dot{x}^\mu 
p^\nu)\\&+\frac{i\kappa}{2}(\dot{x}^\mu+p^\mu)\left(\partial_\nu 
h_\mu^{~\nu}-\frac{1}{2}\partial_\mu 
h\right)+\frac{i\kappa^2}{2}(\dot{x}^\mu+p^\mu)\left(\frac{1}{4}\partial_\mu
(h^{\varrho\sigma}h_{\varrho\sigma})-h^{\nu\sigma}\partial_\sigma 
h_{\mu\nu}\right)\nonumber\\
&+\frac{\kappa^2}{8}\left(\partial_\alpha h^{\alpha\mu}-\frac{1}{2}\partial^\mu 
h\right)^2
\end{align}
Once again, we explicitly isolated the scalar contribution, previously computed 
in \cite{Bonocore:2021qxh}.

The second approximation, which is key in deriving the exponential form of 
the GWL, 
is to perform the soft expansion, defined by the limit $p\gg k$, with $k$ being 
the 
momentum of a soft emission dressing the fermion propagator. In the gauge 
theory case discussed in \cite{Laenen:2008gt, Bonocore:2020xuj} this procedure 
is straightforward, since one can  
solve the path integral order by 
order in the soft expansion, but to 
all-orders in the coupling constant. This must be clearly revisited 
in gravity since, unlike the gauge theory case, the background gravitational 
field is 
intrinsically defined perturbatively via \cref{graviton}. Apart from this 
aspect, the gauge-gravity 
parallel holds. 
We thus rescale $p\rightarrow \lambda p$, $t\rightarrow t/\lambda$, and 
$\kappa\rightarrow \kappa/\lambda$, where $\lambda$ is a bookkeeping parameter (eventually set to one) that
controls the soft expansion in the limit $\lambda \to \infty$. Borrowing a terminology used in  
\cite{Laenen:2008gt}, we denote the leading term as eikonal (E), the subleading 
as next-to-eikonal (NE), and so forth. At NE level, the Lagrangian then 
becomes\footnote{We omit $+\frac{i\kappa}{2\lambda}\dot x^\mu\left(\partial_\nu 
h_\mu^{~~\nu}-\frac{1}{2}\partial_\mu 
h+\partial^{[b}h^{a]}_\mu\psi_a\psi_b\right)$ in \cref{Lexpanded} since the appearance of $\dot{x}$ makes it subleading.} 
\begin{align}
L_{\text{spin}}(t)=&\nonumber\frac{i}{2}\psi\dot{\psi}-\frac{\lambda}{2}\dot{x}^\mu\dot{x}^\nu
 \eta_{\mu\nu}-\frac{\kappa}{2}h_{\mu\nu}(\dot{x}^\mu\dot{x}^\nu+p^\mu 
p^\nu+2\dot{x}^\mu p^\nu)+\\
&\nonumber+\frac{i\kappa}{2\lambda}p^\mu\left(\partial_\nu 
h_\mu^{~~\nu}-\frac{1}{2}\partial_\mu 
h+\partial^{b}h^{a}_\mu\psi_a\psi_b\right)+\mathcal{O}(\lambda^{-2})~\\
=& 
L_{\text{scalar}}(t)+\frac{i}{2}\psi\dot{\psi}+\frac{i\kappa}{2\lambda}p^\mu\partial^{b}h^{a}_\mu\psi_a\psi_b.
\label{Lexpanded}
\end{align}
Note that only the first line of \cref{Lspin} contributes to \cref{Lexpanded}, since the other terms 
are ${\cal O}(\lambda^{-2})$ hence 
sub-subleading.
Recalling that we substituted $x(t)\to x_i+pt+x(t)$ in the path integral representation of the asymptotic propagator,
Feynman rules for the worldline fields $x$ and $\psi$ can then be generated by 
expanding 
$h^{\mu\nu}(x_i+pt+x(t))$ in powers of $x(t)$ around $x=0$. Crucially, as evident in the kinetic terms of \cref{Lexpanded}, the
two-point correlator of $x$ and its time derivatives scale like $1/\lambda$ while the
fermionic correlator scales like $\lambda^0$. Since the fermionic vertex is of order $1/\lambda$, only a finite number of 
diagrams are necessary at any given order in $1/\lambda$. We do not 
repeat the calculation here for the purely bosonic diagrams, which is identical 
to the one presented in 
\cite{Bonocore:2021qxh}. 
 
Spin-dependent diagrams instead follow from one term only, 
namely
\begin{equation}\label{spinvertex}
	\frac{i\kappa}{2\lambda} p^\mu \partial^b h^{a}_{~\mu}\psi_a\psi_b~,
\end{equation}
which generates the following tower of vertices
\begin{align}
\partial^{b}h^{a}_\mu(x_i+pt)\psi_a\psi_b~,\\
x^{\alpha_1}\partial_{\alpha_1}\partial^{b}h^{a}_\mu(x_i+pt)\psi_a\psi_b~,\\
\frac{1}{2}x^{\alpha_1}x^{\alpha_2}\partial_{\alpha_1}\partial_{\alpha_2}\partial^{b}h^{a}_\mu(x_i+pt)\psi_a\psi_b~,
\end{align}
and so forth. 
At NE order no $x$-propagator is required, which means that the path integral 
over $\psi$ is decoupled from the one over $x$. Therefore, to carry out the 
$\psi$ integration it is enough to 
 consider the boundary condition $\psi(0)+\psi(\infty)=2\eta$ and expand
 around the constants $\eta_a$ via
\begin{equation} 
	\psi_a(t)=\eta_a+\tilde{\psi}_a(t)~,
\end{equation}
where the fluctuations now obey the boundary condition
\begin{align}
\tilde{\psi}_a(0)+\lim\limits_{T\rightarrow\infty}\tilde{\psi}_a(T)=0~.
\end{align}
Inserting this expansion into \cref{spinvertex} yields 
\begin{equation}\label{psitilde}
	\frac{i\kappa}{2\lambda} p^\mu \partial^b h^{a}_{~\mu}\left(\eta_a\eta_b+\eta_a\tilde{\psi}_b+\tilde{\psi}_a\eta_b+\tilde{\psi}_a\tilde{\psi}_b\right).
\end{equation}
Notably, $\psi$-fields in \cref{psitilde} have to be contracted at equal 
time, resulting in \\${\braket{\psi_a(t)\psi_b(t)}=\eta_a \eta_b}$. However, thanks to \cref{fermionProp}, the equal-time propagator for 
$\psi$ vanishes. Hence, only the first term of \cref{psitilde} contributes to 
the path integral.
We finally employ the symbol map of \cref{symbolmap}, which gives
\begin{equation} 
	\eta_a\eta_b\mapsto \frac{1}{4}\left[\gamma_a,\gamma_b\right]=-i\sigma_{ab}.
\end{equation}

In conclusion, the contribution from \cref{spinvertex} to the GWL
is given at NE by a single spin-dependent vertex, which modifies the NE 
one-graviton emission. It reads
\begin{equation}\label{spinTerm}
	\frac{i\kappa}{2\lambda}\int\limits_{0}^{\infty}dt ~p^\mu \partial^b 
	h^{a}_{~\mu}\sigma_{ab}~,
\end{equation}
or, in momentum space,
\begin{equation}\label{spinTerm2}
	\int 
	\frac{d^4k}{(2\pi)^4}\tilde{h}_{\mu\nu}\left[\frac{i\kappa}{2}\frac{p^\mu 
	k_\varrho}{kp}\sigma^{\varrho \nu}\right]~.
\end{equation}
 Therefore,
combining the vertex of \cref{spinTerm} 
with the scalar vertices derived in \cite{Bonocore:2021qxh}, we obtain the 
following representation for 
the spin 1/2 GWL in curved space
\begin{align}\label{eq:GWL_spin}
\widetilde W_p(0,&\infty ;x_i)=\exp\Bigg\{
\frac{i\kappa}{2}\!\int\limits_0^{\infty}\! dt\,\left[\!-\! p_{\mu}p_{\nu}
\!+\!ip_{\nu}\partial_{\mu}
\!-\!\frac{i}{2}\eta_{\mu\nu}p^{\alpha}\partial_{\alpha}
\!+\!\frac{i}{2}tp_{\mu}p_{\nu}\partial^2\!+\!\sigma^{\nu\sigma}p^\mu\partial_\sigma
 \right]h^{\mu\nu}(x_i\!+\!pt)
\notag \\
&\! + \frac{i\kappa^2}{2}\int\limits_0^{\infty}\! dt  \!\int\limits_0^{\infty}\! ds\,
\Bigg[
\frac{p^{\mu}p^{\nu}p^{\rho}p^{\sigma}}{4}\min(t,s)\,\partial_{\alpha}
h_{\mu\nu}(x_i+pt)
\partial^{\alpha}
h_{\rho\sigma}(x_i+ps)+
\notag\\&
+p^{\mu}p^{\nu}p^{\rho}\,\theta(t\!-\!s)\,
h_{\rho\sigma}(x_i\!+\!ps)
\partial_{\sigma}
h_{\mu\nu}(x_i\!+\!pt)
\!+\!p^{\nu}p^{\sigma}\,\delta(t\!-\!s)\,
h^{\mu}_{\;\,\sigma}(x_i\!+\!ps)
h_{\mu\nu}(x_i\!+\!pt)
\Bigg]\Bigg\},	
\end{align} 
where the arguments $0$ and $\infty$ refer to the boundaries of the time integrations, which have been made explicit for upcoming generalizations.
Eq.~(\ref{eq:GWL_spin}) is one of the main results of this paper. A few comments 
are 
in order. 

The spin dependent term in \cref{spinTerm} and \cref{spinTerm2}, which appears 
in the first line of 
\cref{eq:GWL_spin}, is the only modification w.r.t. the 
scalar GWL. It
represents the universal soft factor which arises in the factorization of 
a single soft emission in a 
scattering amplitude, in agreement with single-emission next-to-soft theorems 
\cite{Cachazo:2014fwa}. Note, however, that the result 
of \cref{eq:GWL_spin} is more general, since it involves correlations between two 
gravitons. It shows in fact that no 
	spin-dependent modification 
	is necessary 
	for two-graviton 
emissions. This aspect is analogous to the 
gauge 
theory case, where the spin dependent contribution in the double-emission is 
proportional to the commutator $\sigma_{\mu\nu}[A^\mu,A^\nu]$ and thus vanishes 
in the abelian 
limit. Therefore, the 
absence of spin dependent terms in two-graviton emissions is due to the 
abelian nature of the graviton field. Note however that at sub-subleading 
(i.e. NNE) level correlations among 
three gravitons 
need to be taken into account. 

Secondly, we note the presence of an explicit dependence of the GWL on $x_i$, 
which was introduced as an integration constant while solving the free equations of 
motion in flat space. Its effect is to shift the starting point of the particle 
trajectory away from the origin and is due to 
finite-size effects of the hard interaction. 
This NE effect captures (part of) the orbital contribution of the 
soft theorem for gravitons and in the case of classical scattering becomes the impact parameter. 
Indeed, as discussed in \cite{White:2011yy, Bonocore:2021qxh}, the
presence of $x_i\neq 0$ mixes with derivatives of the hard interaction to 
generate the full orbital angular momentum. 

The soft GWL in \cref{eq:GWL_spin} manifestly demonstrates
the exponentiation of subleading soft gravitons dressing the external fermionic state of a scattering amplitude. 
In the next section we discuss
how \cref{eq:GWL_spin} 
has to be modified in the
classical limit. As we shall see, the parallel with the
soft exponentiation discussed in this section will turn out to be a powerful tool in this derivation. 

\section{Classical scattering and GWLs} \label{sec_ampltidues_and_classical}

As shown in \cite{Bonocore:2021qxh} for the case of non-spinning objects in the 
conservative sector, classical observables in the Post Minkowskian (PM) 
expansion (i.e. the expansion in $G$) can be computed via vacuum expectation values of classical GWLs.
We now extend this method to spin $1/2$ particles, thus allowing 
the inclusion of perturbative corrections which are linear the classical spin variable $S$.
 In doing so, we clarify some aspects concerning the construction of 
classical GWLs originally proposed in \cite{Bonocore:2021qxh}.

\subsection{Regge VS soft VS classical}
The starting point of the GWL approach for classical scattering is to consider 
the 
Regge limit 
of the (quantum) $2\to 2$ scattering amplitude, where the Mandelstam variables 
fulfill\footnote{Note also that we assume masses to be much larger than $t$.} 
$s\gg t$ 
and the exchanged gravitons 
are soft.  More 
specifically, observables at order $\kappa^{2n}$ (i.e. $nPM$) require GWLs expanded to order $\kappa^n$, which in turn 
are
constructed via the soft expansion at the $N^{n-1}E$ level.

Specifically, for scattering of spin $1/2$ objects, we consider the following 
next-to-eikonal amplitude
\begin{align}
   &\mathcal{A_{\text{E+NE}}}\nonumber\\=&\bra{0}\bar{u}(p_3)
   \widetilde{W}_{p_3}(0,\infty;b)\widetilde{W}_{p_1}(-\infty,0;b)u(p_1) 
   \bar{u}(p_4)\widetilde{W}_{p_4}(0,\infty;0)
   \widetilde{W}_{p_2}(-\infty,0;0)u(p_4)\ket{0},
   \label{regge}
\end{align}
where the definition for the GWL starting from $-\infty$ to $0$ can be 
trivially obtained from \cref{eq:GWL_spin} by flipping the boundary conditions.
The dependence over $b$ in the GWL with momenta $p_1$ and $p_3$ instead 
represents a 
constant overall shift which separates the trajectories of the hard particles 
and therefore plays the role of the impact parameter\footnote{It can be shown 
that due to worldline supersymmetry, one can exploit a gauge symmetry among the 
variables $p_i$, $b$ and $\eta$, to eliminate two components of $b$,
 choosing it for example to be perpendicular to the incoming momenta 
 \cite{Jakobsen:2021zvh}. }. 
 
Note that by including subleading soft effects with GWLs, in \cref{regge} we 
are in fact including corrections to the strict Regge limit $t/s\to 0$.
However, \cref{regge} contains quantum effects that one wishes to discard. In 
order for the method to be efficient for the calculation of classical 
observables, 
it would be desirable to isolate the classical contributions at the integrand 
level, rather than taking $\hbar\to 0$ of the final result, so that 
classical observables can be computed from the exponential form of \cref{regge}, i.e.
\begin{align}
 \mathcal{A_{\text{E+NE}}}\overset{\hbar \to 0}{\longrightarrow} e^{i(\chi_{\text{E}}+\chi_{\text{NE}})}~,
\end{align}
where $\chi_{\text{(N)E}}$ is called (next-to-)eikonal phase. 
In fact, the classical limit 
in this approach corresponds to
 a strict Regge limit, while subleading Regge corrections contribute to the 
 so-called Regge trajectory of the graviton and are quantum.
 Accordingly, in order to efficiently compute classical observables through a 
 VEV of GWL 
 operators, one has to  
 replace the GWL defined in \cref{eq:GWL_spin} with a properly-defined 
 classical 
GWL. 

This approach has been carried out in the scalar case \cite{Bonocore:2021qxh}. 
By reinstating $\hbar$, indeed, one can identify the terms in the worldline 
Lagrangian (and subsequently in the exponentiated vertices of the GWL) that 
vanishes in the limit $\hbar\to0$. Specifically, NE single emissions in the GWL 
correspond to a small recoil of the hard particle (hence a subleading 
Regge effect), which are necessarily quantum 
phenomena. Double-graviton vertices, on the other 
hand, survive the classical limit. In this way, the 
soft expansion can be interpreted in classical sense as a sum over correlations 
among soft emissions rather than an expansion in the 
energy of a single emission. 

However, in order to properly define a classical GWL, 
there are still two aspects to be addressed. 
 The first one is that, for spin variables, a naive $\hbar$ power counting is 
 not suitable for the classical limit, since the spin of a point-like particle 
 is proportional to 
 $\hbar$, hence intrinsically a quantum property. We rather need a map to 
 classical quantities. 
 The second one is that the 
 boundary conditions of the worldline path integral in the GWL must be modified. Although most of these 
 aspects have already been addressed in the recent literature (see 
 \cite{Kosower:2018adc, Maybee:2019jus, Jakobsen:2021zvh, Bonocore:2021qxh}) we 
 would like to 
 discuss them in greater detail to elucidate their role in the construction of a 
 classical GWL. 
  We discuss these aspects in 
 turn.

\subsection{Classical limit and spin}
As already mentioned, to extract the classical limit one 
has to compute the $\hbar$ scaling of different terms in the worldline 
Lagrangian. 
Specifically, momenta of external particles are regarded as classical momenta, 
but the momenta of massless force carriers are rescaled as $k=\hbar\bar{k}$ 
\cite{Kosower:2018adc}. 
As shown in \cite{Bonocore:2021qxh}, this power-counting analysis is enough
to identify the classical terms in the scalar GWLs that are necessary for 
classical observables.  
 In the case of spin, however, we first need to
 select the relevant classical variable for the spin of an extended body in GR, 
 i.e. the classical spin tensor $S^{\mu\nu}$, defined as the canonical momentum 
 conjugated to the angular velocity $\Omega_{\mu\nu}$.
The obvious choice for the spin tensor is the spin Lorentz generator 
$\sigma^{\mu\nu}$. Recalling that on the worldline path integral this generator 
naturally emerges from the symbol map $i\eta^a\eta^b\rightarrow \sigma^{ab}$, 
where $\eta_a$ represent the boundary condition for the integral over 
$\psi_a$,  
 we can make the identification 
\begin{equation}
i\eta^a\eta^b\equiv S^{ab}~, 
\end{equation}
which establishes a map\footnote{This map facilitates neglecting the spinors in the definition of $\mathcal{A}_{\text{NE}}$, as it is compatible with the usual identification $u(p)\sigma^{\mu\nu}\bar{u}(p')=S^{\mu\nu}u(p)\bar{u}(p')+{\cal O}(\hbar)$ used in traditional amplitude approaches.} between the worldline spin variable $\psi_a$ and the 
classical spin tensor $S^{ab}$.

To determine the scaling of the Grassmann fields $\psi_a$ in the classical 
limit we then
consider the following argument. In the rest frame of the particle the spatial 
components of the Lorentz generator can be arranged into the spin operator, 
whose eigenvalue scales as $\hbar s$, which suggests that the quantum 
spin scales as $\psi^a\psi^b\sim \hbar$. For the spin to survive the classical 
limit, we therefore need to consider the limit in which $s\rightarrow \infty$ 
in such a way that the spin remains a macroscopic variable. The standard way to 
achieve this with amplitude methods is by rescaling $\psi_i^a\rightarrow 
\frac{1}{\sqrt{\hbar}}\psi_i^a$. As we are going to discuss, in this work 
we implement it instead by combining the classical limit of the spin degrees of freedom with the 
soft limit.

We first observe that loops of the worldline 
fields $x$ and $\psi$ contribute only to the quantum part of the asymptotic 
propagator. The reason for that is fairly straightforward and was given in 
\cite{Mogull:2020sak}, namely that the path integral should solve the classical 
equations of motion for $x$ and $\psi$, which only requires tree level graphs. 
This can also be seen by power counting arguments. There is only one 
term in the Lagrangian that would allow for an $x$-loop without expanding 
$h_{\mu\nu}$, but this term diverges and is cancelled by the ghost 
contributions 
(see \cite{Bonocore:2021qxh} for a detailed discussion). This means for every 
other $x$ loop it is required to expand $h_{\mu\nu}$ into a power series 
introducing extra derivatives that add additional powers of $k$ and therefore 
of $\hbar$.

The argument for $\psi$ loops works similarly. Upon implementing the boundary 
conditions for $\psi$ we expand 
\begin{equation}\label{psiexp}
\psi_a \psi_b=(\eta_a+\tilde{\psi}_a)(\eta_b+\tilde{\psi}_b).
\end{equation} 
This tells us that the $\psi$-vertices are divided into three categories: one 
with all constant Grassmann vectors, which get mapped onto the spin tensor, one 
with all dynamical fields, and one with every possible mixture of the two. If a 
worldline diagram contains a $\tilde{\psi}$-loop the propagators remove one 
spin tensor more as in the corresponding loop-free diagram. Therefore, the loop 
diagram is subleading in $\hbar$ and thus non classical.
The case of mixed loops is also straightforward. Due to the expansion of the 
Grassmann vectors into the constant part $\eta$ plus a dynamical field, every 
diagram with an $x$-$\tilde{\psi}$-loop also exists as an $x$-propagator 
tree-level diagram. The $\tilde{\psi}$-propagator that would turn this diagram 
into a loop diagram thus removes a potential spin tensor and the diagram 
contributes only to the quantum part of the asymptotic propagator.
We should thus neglect all loops involving $\psi$. This can be elegantly achieved
by slightly altering the power counting in $\lambda$ introduced in \cref{Lexpanded} for
the soft expansion. For the classical GWL we additionally
rescale $\psi\rightarrow 
\sqrt{\lambda} \psi$. In this way, $\lambda$ suppresses both $x$ and $\psi$ loops, eliminating all quantum correction from the soft GWL.

Further we note that in the derivation of the soft GWL 
we 
rescaled also $\kappa\rightarrow \kappa/\lambda$, which resulted in a 
correspondence between the soft expansion and the $\lambda$-expansion. Such 
correspondence is independent of the order in the PM expansion. Specifically, 
the eikonal approximation is order $\lambda^0$, NE approximation is order 
$\lambda^{-1}$, and so forth. However, for the classical limit it is actually 
more convenient not to rescale the coupling $\kappa$. Although this choice 
spoils the simple relation between the soft and the $\lambda$ expansion, 
the classical terms emerge neatly as those of order $\lambda$.
Therefore, we conclude that the relevant Lagrangian in the classical limit is 
\begin{align}
L^{\text{cl.}}_{\text{spin}}=&\frac{\lambda}{2}\left(i\psi_a\dot{\psi}^a-\dot{x}^\mu\dot{x}^\nu
 \eta_{\mu\nu}-\kappa h_{\mu\nu}(\dot{x}^\mu\dot{x}^\nu+p^\mu 
p^\nu+2\dot{x}^\mu p^\nu)+i\kappa 
p^\mu\left(\partial^{[b}h^{a]}_\mu\psi_a\psi_b\right)\right)\!+\!\mathcal{O}(\lambda^{0}),
\label{Lexpandedclass}
\end{align}
where the graviton field is evaluated at $h_{\mu\nu}=h_{\mu\nu}(x_i+pt+x(t))$.

\subsection{The classical GWL}
After having identified the terms in the worldline Lagrangian that are 
relevant in the classical limit, we need to consider the boundary conditions of 
the path integral that leads to a proper definition of the GWL. 
In this regard, we first observe that when we derived the soft GWL in 
\cref{sec_worldlineexp} from a localized 
hard interaction to an asymptotic state of definite momentum, we considered the 
following path integral 
\begin{equation} 
	\widetilde{W}_p(0,\infty;x_i)=\int \mathcal{D}x\mathcal{D}\psi \exp\left\{i 
	\int_{0}^{\infty}L[x(t),\psi(t)]dt\right\}~.
\end{equation}
We solved the path integration perturbatively 
by using the boundary conditions $x(0)=0$ and $\psi(0)+\psi(\infty)=0$. Then, in
order describe factorized amplitudes at next-to-soft level in terms of GWLs and construct a next-to-soft function, we attached
each GWL, corresponding to an external state of hard momentum, to the hard function $\mathcal{H}$ at $t=0$.
In classical settings however, we want the incoming and outgoing 
particles to be causally connected through an uninterrupted matter line, as 
evident in the strict Regge limit description, which essentially sets $p_1=p_3$ and $p_2=p_4$. As we shall see, gluing together
two Wilson lines from $-\infty$ to $0$ and from $0$ to $+\infty$ is not the most efficient way to define a classical GWL.

There is a clear diagrammatic interpretation for this. The boundary conditions 
at $t=0$ cause the diagram to pick up a pole that originates from the 
matter propagator connecting the hard interaction to the first soft graviton emission
from that line. Due to these boundary conditions, naively multiplying two GWLs does not give the correct pole structure in the correlated multi-graviton exchange diagrams (i.e. those involving the vertices in the last two lines of \cref{eq:GWL_spin}). For instance, in the two-graviton emission diagram this can be viewed as the particle not fulfilling its equations of motion at $t=0$. As this is clearly a non-classical 
phenomenon, we need to discard it. Discarding these contributions would need to be considered part of the gluing process to obtain a classical result from $\mathcal{A}_{\text{NE}}$.

To restore the classical history of the worldline particle from 
$-\infty$ to $+\infty$, it is in fact more straightforward not to glue the GWLs but 
rather the 
actions for incoming and outgoing particles. In this way the path integral solves for the classical trajectory from the far past to the far future and not for the incoming and outgoing paths separately, in analogy with 
\cite{Mogull:2020sak}. This leads to the following definition for the classical GWL 
\begin{equation}\label{gluingactions}
	\widetilde{W}^{\textbf{cl.}}_{p}(x_i)=\int  \mathcal{D}x\mathcal{D}\psi \exp\left\{i\int_{-\infty}^{\infty}L^{\text{cl.}}_{\text{spin}}[x(t),\psi(t)]dt\right\},
\end{equation} 
where $L^{\text{cl.}}_{\text{spin}}$ is given by \cref{Lexpandedclass} and the dynamical fields have been extended to the entire real line $(-\infty,\infty)$.
Specifically, to better appreciate the 
relation with the boundary condition of the soft GWL, we decompose the action in \cref{gluingactions} into the sum of two copies of the soft action with flipped boundary conditions, i.e.  
\begin{equation}
    \int_{-\infty}^\infty L[x(t),\psi(t)]=\int_0^\infty (L_+[x_+(t),\psi_+(t)]+L_-[x_-(t),i\psi_-(t)])dt,
\end{equation}
where $L_-$ is related to $L_+$ simply by replacing $p\to -p$ and $\eta\to i\eta$\footnote{The sign flip in the spin tensor can be achieved by multiplying the Grassmann fields by a factor of $i$.}. The reader may now check that the Lagrangian is build such that this sign flip is compensated exactly by substituting $t\rightarrow -t$. This allows us to combine the Lagrangians after defining $x(t)=\theta(t)x_+(t)+\theta(-t)x_-(-t)$ (and similarly for $\psi$)
representing one uninterrupted matter line, where the worldline fields are 
allowed to fluctuate across the previous boundary at $t=0$. 

With this definition, one has to re-derive new worldline Green functions which 
are compatible with the $(-\infty,+\infty)$ boundary conditions. In fact, for the soft GWL we 
used 
\begin{equation} 
	\braket{x^\mu(t)x^\nu(s)}=i\eta^{\mu\nu} G^{\textbf{soft}}(t,s)\text{\qquad where\qquad}G^{\textbf{soft}}(t,s)=-\min(t,s).
\end{equation}
Here instead we need to use a classical Green function $G^{\textbf{cl.}}$ 
defined by
\begin{equation} 
	\braket{x^\mu(t)x^\nu(s)}=i\eta^{\mu\nu} G^{\textbf{cl.}}(t,s)\text{\qquad 
	where\qquad}G^{\textbf{cl.}}(t,s)=\frac{|t-s|}{2}~.
\end{equation}
Note that the soft and classical Green functions are related via 
\begin{align}
2G^{\textbf{cl.}}=G^{\textbf{soft}}(t,s)+G^{\textbf{soft}}(-t,-s)~.
\label{softcl}
\end{align}
The Fourier representation might at times be more helpful. It reads 
\begin{equation} 
	\tilde{G}(\omega)=-\frac{1}{2}\left(\frac{1}{(\omega+i\epsilon)^2}
	+\frac{1}{(\omega-i\epsilon)^2}\right)~,
\end{equation}
where $G(t-s)=\int\frac{d\omega}{2\pi} \tilde{G}e^{-i\omega(t-s)}$.
In Fourier space it is also more natural to find the derivatives of the Green 
functions. Finally, note that the Green function of the fermionic field 
$\psi(t)$, which in the soft calculation reads
\begin{equation} 
	\braket{\psi^a(t)\psi^b(s)}=\frac{1}{2}\eta^{ab}(\theta(t-s)-\theta(s-t))=\eta^{ab}\cal G^{\textbf{cl.}}~,
\end{equation} 
needs no modifications.

 We can thus solve the path integral for the classical GWL and from there 
 extract the vertices.
We start with the Lagrangian in \cref{Lexpandedclass} and expand the 
Grassmann variables over the background via \cref{psiexp}, to obtain
\begin{align} 
	L^{\text{cl.}}_{\text{spin}}=&\frac{i\lambda}{2}
	\psi\dot{\psi}-\frac{\lambda}{2}\dot{x}^\mu\dot{x}^\nu 
	\eta_{\mu\nu}-\frac{\kappa\lambda}{2}h_{\mu\nu}(\dot{x}^\mu\dot{x}^\nu+p^\mu
	 p^\nu+2\dot{x}^\mu p^\nu)\nonumber\\
	 &+\lambda\frac{i}{2}(p^\mu+\dot{x}^\mu) 
	\omega_\mu^{ab}(\eta_a\eta_b+2\eta_a\psi_b+\psi_a\psi_b)~,
\end{align}
where we have neglected non-classical terms. Up to order $\kappa^2$ and setting $\lambda=1$, we need the following worldline vertices\footnote{The $\mathcal{O}(\kappa)$ expansion of $\omega_\mu^{ab}$ is left implicit to avoid clutter.} 
\begin{align}
	\text{\textcircled{0}}=&\frac{i}{2}\int\limits_{-\infty}^\infty d\tau\left(-\kappa p^\mu p^\nu h_{\mu\nu}(x_i+p\tau)+p^\mu\omega^{ab}_\mu(x_i+p\tau) S_{ab}\right)\\
	\text{\textcircled{1}}=&\frac{i}{2}\int \limits_{-\infty}^\infty d\tau\left(-\kappa p^\mu p^\nu \partial_\alpha h_{\mu\nu}(x_i+p\tau)+p^\mu\partial_\alpha\omega^{ab}_\mu(x_i+p\tau) S_{ab}\right)x^\alpha(\tau)\\
	\text{\textcircled{2}}=&\frac{i}{2}\int\limits_{-\infty}^\infty d\tau \left(-2\kappa p^\nu h_{\mu\nu}(x_i+p\tau)+\omega_\mu^{ab}(x_i+p\tau)S_{ab}\right)\dot{x}^\mu(\tau)\\
	\text{\textcircled{3}}=&i\int\limits_{-\infty}^\infty d\tau 
	\left(ip^\mu\omega_\mu^{ad}(x_i+p\tau)\eta_a\right)\psi_d(\tau)~.
\end{align}
They assemble into the classical GWL using the usual exponentiation of 
connected diagrams in QFT, i.e. schematically
\begin{equation} 
\widetilde{W}^{\textbf{cl.}}
=\exp\left\{\text{\textcircled{0}}+\frac{1}{2}
\braket{\text{\textcircled{1}}\!\!-\!\!\text{\textcircled{1}}
	+\text{\textcircled{2}}\!\!-\!\!\text{\textcircled{2}}
	+\text{\textcircled{3}}\!\!-\!\!\text{\textcircled{3}}
	+2\text{\textcircled{1}}\!\!-\!\!\text{\textcircled{2}}}\right\}~.	
\end{equation} 
Explicitly, we obtain in position space
\begin{align}
\nonumber\widetilde{W}_p^{\textbf{cl.}}(x_i)=\exp\Bigg\{&\!-\frac{i\kappa}{2}\!\!\int\limits_{-\infty}^\infty \!
 d\tau_1 \left[p^\mu p^\nu-p^\mu 
S^{\nu\varrho}\partial_\varrho\right]h^1_{\mu\nu}-\frac{i\kappa^2}{8}\!\int\limits_{-\infty}^\infty
\! d\tau_1d\tau_2 \times\\&\nonumber\times\Bigg[p^\mu p^\nu p^\varrho p^\sigma 
h^1_{\mu\nu,\alpha}h^{2~,\alpha}_{\varrho\sigma}G^{\textbf{cl.}}+4p^\nu 
p^\sigma\eta^{\mu\varrho}h^1_{\mu\nu}h^2_{\varrho\sigma} 
\ddot{G}^{\textbf{cl.}}+4p^\mu p^\nu p^\sigma 
h^1_{\mu\nu,\varrho}h^{2~\varrho}_{~\sigma}\dot{G}^{\textbf{cl.}}\\
&\nonumber+S_{ab}\Bigg(-4p^\mu p^\varrho h^1_{\mu [a,d]}h^2_{\varrho[b,d]} 
{\cal G}^{\textbf{cl.}}-2p^\mu p^\nu p^\varrho h^1_{\mu 
\nu,\alpha}h^{2~,\alpha}_{\varrho[a,b]}G^{\textbf{cl.}}\\\nonumber
&+2\left(p^\mu p^\nu h^1_{\mu\nu,\varrho}h^{2\varrho[a,b]}+2p^\mu 
p^\varrho 
h^1_{\mu 
[a,b],\sigma}h^{2\varrho\sigma}\right)\dot{G}^{\textbf{cl.}}-4p_\nu 
h^{\mu\nu}_1h^2_{\mu [a,b]}\ddot{G}^{\textbf{cl.}}\\&+2p^\mu\left(
h^{\varrho 
	a}_1h^{2~b}_{\mu,\varrho}+h_1^{b\varrho}h_{\mu\varrho}^{2~,a}+\frac{1}{2}h_1^{b\varrho}h^{2~a}_{\varrho~,\mu}\right)\ddot{G}^{\textbf{cl.}}
\Bigg)\Bigg]\Bigg\}~,\label{clGWL}
\end{align}
where we are using the 
following shorthand notation 
$h^1_{\mu\nu}=h_{\mu\nu}(p\tau_1+x_i)$, 
$h^2_{\mu\nu}=h_{\mu\nu}(p\tau_2+x_i)$, 
$\dot{G}^{\textbf{cl.}}=\frac{\partial}{\partial t}G^{\textbf{cl.}}$. We also 
used   
$\ddot{G}^{\textbf{cl.}}=\frac{\partial^2}{\partial t\partial 
	s}G^{\textbf{cl.}}$. 
Alternatively, decomposing the graviton into its momentum modes allows us to solve the time integrals, yielding 
\begin{align}
\nonumber\widetilde{W}_p^{\textbf{cl.}}(x_i)=\exp&\Bigg\{\!-\frac{i\kappa}{2}\!\int_k 
\tilde{h}_{\mu\nu}(k)e^{-ikx_i}\hat{\delta}(pk)\left[p^\mu p^\nu+ip^\mu (S\cdot 
k)^\nu\right]\\
&\nonumber-\!\frac{i\kappa^2}{8}\!\int_{k,l}\frac{\tilde{h}_{\mu\nu}(k)\tilde{h}_{\varrho\sigma}(l)}{2}e^{-i(k+l)x_i}\Bigg[-2I_1 (kl)p^\mu p^\nu p^\sigma( p^\varrho+2i(S\cdot l)^\varrho)
\\
&\nonumber\!-4I_2p^\nu k^\varrho\big(\!-2ip^\mu  p^\sigma +p^\mu  (S\cdot l)^\sigma+2
p^\sigma (S\cdot k)^\mu\big)+8I_4 p^\mu 
p^\varrho S_{ab} 
k^{[c}\eta^{a]\nu}l^{[c}\eta^{b]\sigma}\\&+8I_3\eta^{\mu\varrho}(p^\nu p^\sigma 
+ip^\nu (S\cdot l)^\sigma)\Bigg]\Bigg\}~,\label{clGWLmom}
\end{align}
where $(S\cdot k)^\sigma=S^{\sigma\nu}k_\nu$, we used the shorthand notation 
$\int \frac{d^D k}{(2\pi)^D}=\int_k$ and we defined
\begin{align}
I_1&=\int\limits_{-\infty}^\infty dt\, ds\, 
G^{\textbf{cl.}}(t,s)e^{-ipkt-ipls}=-2\hat{\delta}(p(k+l))\left(\frac{1}{(p(k-l)+i\epsilon)^2}+\frac{1}{(p(k-l)-i\epsilon)^2}\right)\\
I_2&=\int\limits_{-\infty}^\infty dt\, ds \,
\dot{G}^{\textbf{cl.}}(t,s)e^{-ipkt-ipls}=-i\hat{\delta}(p(k+l))\left(\frac{1}{(p(k-l)+i\epsilon)}+\frac{1}{(p(k-l)-i\epsilon)}\right)\\
I_3&=\int\limits_{-\infty}^\infty dt\, ds \,
\ddot{G}^{\textbf{cl.}}(t,s)e^{-ipkt-ipls}=-\hat{\delta}(p(k+l))\\
I_4&=\int\limits_{-\infty}^\infty dt\, ds \,
{\cal G}^{\textbf{cl.}}(t,s)e^{-ipkt-ipls}=I_2
\end{align}
Equations (\ref{clGWL}) and (\ref{clGWLmom}) form one the main result of this 
paper. 

At this point it is instructive to compare \cref{clGWL} and \cref{clGWLmom} 
with the soft GWL of \cref{eq:GWL_spin}. The single-graviton terms have an easy 
correspondence. Specifically, both the scalar and the spin terms in the first 
line of 
\cref{clGWLmom} correspond to the sum of two 
soft emissions from $-\infty$ to $0$ and from $0$ to $+\infty$, respectively. 
These can be identified in \cref{clGWLmom} by neglecting purely quantum recoil 
terms and (in momentum space) using
\begin{align}
2\pi i\,\delta(p\cdot k)
=\frac{1}{p\cdot k+i \epsilon}-\frac{1}{p\cdot k-i \epsilon}~.
\end{align}
Note, however, that while the spin term is a subleading effect in the soft GWL, 
it 
becomes a leading effect in the classical limit. The reason is that while 
in the purely soft limit $k\cdot S$ is subleading w.r.t. terms with no powers 
of the soft momentum 
$k$, the scaling of the spin variable $S$ is enhanced in the classical limit, 
such that 
$k\cdot S={\cal O}(1)$.
Similarly, the second and third line of \cref{clGWL} can be mapped to \cref{eq:GWL_spin}  
via \cref{softcl}, which again corresponds to a sum of soft emissions from 
$-\infty$ to $0$ and from $0$ to $+\infty$, respectively. 
On top of that, we note an explicit dependence over spin for two-graviton 
emissions, which has no equivalent 
in \cref{eq:GWL_spin}. The reason is again the enhancement in the scaling of 
the spin variable. 
We conclude that the soft and the classical GWLs have the same structure, with the only differences due to the presence of different Green functions (because of the boundary conditions) and different spin dependence (because of different scaling in the classical limit). 

\subsection{Observables at 2PM}
Equipped with the classical GWL derived in the previous section, we can now 
discuss how classical observables are computed. Specifically, if we limit the
analysis 
to the conservative sector up 
to 2PM, the strategy is the same as the one outlined in the scalar case 
\cite{Bonocore:2021qxh}. 
In particular, classical observables are generated from the (next-to-)eikonal 
phase, 
which for large impact parameter $b$ can be computed as the vacuum expectation 
value (VEV) of two classical GWLs, i.e.
\begin{align}
e^{i (\chi_{\text{E}}+\chi_{\text{NE}})}=
\langle
0|
\widetilde{W}_{m_1u_1}^{\textbf{cl.}}(0) 
\widetilde{W}_{m_2u_2}^{\textbf{cl.}}(b)
|0
\rangle~,
\label{VEV}
\end{align}
where we have introduced the velocities $u_i^\mu$ via $p_i^\mu=m_iu_i^\mu$. 
Scalar and spin observables are then obtained after differentiating $\chi$ 
w.r.t. the impact parameter $b$. The VEV in \cref{VEV} 
is computed by using the standard (gauge-fixed) Einstein-Hilbert action. 
Diagrams are thus generated perturbatively by both the vertices contained in 
each GWL and by the vertices in the bulk. For sake of completeness, we  
list all relevant Feynman rules in 
\cref{sec:feynman}.

At 1PM we can then assemble the Feynman rules from the GWL into three 
diagrams, as shown in \cref{fig:1pm}. 
The computation of the three diagrams is presented in \cref{sec:1pm}. The final 
result reads
\begin{equation}\label{1pmscal}
i\chi_{\text{E},0}=i\frac{\kappa^2m_1m_2}{16\pi} 
\frac{1}{\sqrt{\gamma^2-1}}\Gamma(-\epsilon)\left(\gamma^2-\frac{1}{2}\right)(-b^2)^\epsilon~,
\end{equation}
and
\begin{equation}\label{1pm}
i\chi_{\text{E},1}=i\frac{\kappa^2m_1m_2}{16\pi} 
\frac{1}{\sqrt{\gamma^2-1}}(-b^2)^{\epsilon-1}\gamma
b_\mu(S_1^{\mu\nu}u_{2\nu}-S_2^{\mu\nu}u_{1\nu})~,
\end{equation}
where $\gamma=u_1\cdot u_2$ and we denoted with $\chi_{\text{E},i}$ terms of 
order $S^i$. The result is in agreement with the 
literature (see e.g.
\cite{DiVecchia:2020ymx, Damour:2020tta, Aoude:2020ygw, Bern:2021dqo, Jakobsen:2021zvh}). 

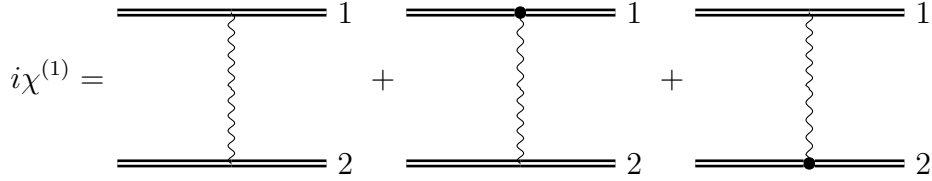
\begin{figure}[thb]
	\begin{align*} 
	i\chi^{(1)}=\begin{tikzpicture}[baseline=(c1)]
	\begin{feynman}
	\vertex (c1) at(0,0);
	\vertex (a) at (-1.5,1);
	\vertex (b) at (0,1) ;
	\vertex (c) at (1.5,1){\(1\)};
	\vertex (d) at (-1.5,-1);
	\vertex (e) at (0,-1) ;
	\vertex (f) at (1.5,-1){\(2\)};
	\diagram*{(a)--[line width=1pt,double distance=1pt](b)--[line 
	width=1pt,double 
		distance=1pt](c),(d)--[line width=1pt,double distance=1pt](e)--[line 
		width=1pt,double distance=1pt](f),(b)--[photon](c1)--[photon](e)};
	\end{feynman}
	\end{tikzpicture}+\begin{tikzpicture}[baseline=(c1)]
	\begin{feynman}
	\vertex (c1) at(0,0);
	\vertex (a) at (-1.5,1);
	\vertex (b)[dot] at (0,1) {};
	\vertex (c) at (1.5,1){\(1\)};
	\vertex (d) at (-1.5,-1);
	\vertex (e) at (0,-1) ;
	\vertex (f) at (1.5,-1){\(2\)};
	\diagram*{(a)--[line width=1pt,double distance=1pt](b)--[line 
	width=1pt,double 
		distance=1pt](c),(d)--[line width=1pt,double distance=1pt](e)--[line 
		width=1pt,double distance=1pt](f),(b)--[photon](c1)--[photon](e)};
	\end{feynman}
	\end{tikzpicture}+\begin{tikzpicture}[baseline=(c1)]
	\begin{feynman}
	\vertex (c1) at(0,0);
	\vertex (a) at (-1.5,1);
	\vertex (b) at (0,1) ;
	\vertex (c) at (1.5,1){\(1\)};
	\vertex (d) at (-1.5,-1);
	\vertex (e)[dot] at (0,-1) {};
	\vertex (f) at (1.5,-1){\(2\)};
	\diagram*{(a)--[line width=1pt,double distance=1pt](b)--[line 
	width=1pt,double 
		distance=1pt](c),(d)--[line width=1pt,double distance=1pt](e)--[line 
		width=1pt,double distance=1pt](f),(b)--[photon](c1)--[photon](e)};
	\end{feynman}
	\end{tikzpicture}
	\end{align*}
	\caption{Diagrams contributing to the eikonal phase at 1PM. Dots represent 
	spin vertices.}
	\label{fig:1pm}
\end{figure}

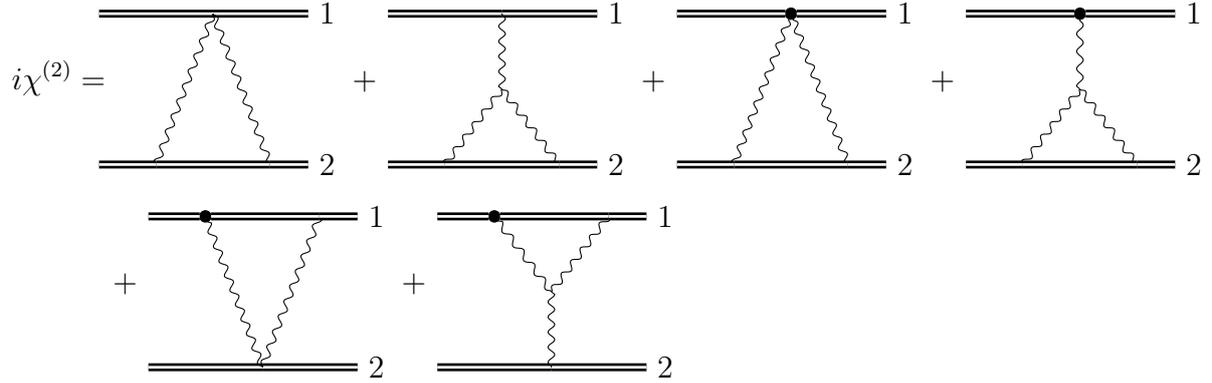
\begin{figure}[thb]
	\begin{align*} 
	i\chi^{(2)}=&\!\!\begin{tikzpicture}[baseline=(c1)]
	\begin{feynman}
	\vertex (c1) at (0,0);
	\vertex (u1) at (-1.5,1);
	\vertex (u2) at (0,1);
	\vertex (u3) at (1.5,1) {\(1\)};
	\vertex (l1) at (-1.5,-1);
	\vertex (l2) at (-0.75,-1);
	\vertex (l3) at (0.75,-1);
	\vertex (l4) at (1.5,-1) {\(2\)};
	\diagram*{(u1)--[line width=1pt,double distance=1pt](u2)--[line 
		width=1pt,double 
		distance=1pt](u3),u2--[photon](l2),u2--[photon](l3),(l1)--[line
		width=1pt,double distance=1pt](l2)--[line width=1pt,double 
		distance=1pt](l3)--[line width=1pt,double distance=1pt](l4)};
	\end{feynman}
	\end{tikzpicture}
	+\begin{tikzpicture}[baseline=(c1)]
	\begin{feynman}
	\vertex (c1) at (0,0);
	\vertex (u1) at (-1.5,1);
	\vertex (u2) at (0,1);
	\vertex (u3) at (1.5,1) {\(1\)};
	\vertex (l1) at (-1.5,-1);
	\vertex (l2) at (-0.75,-1);
	\vertex (l3) at (0.75,-1);
	\vertex (l4) at (1.5,-1) {\(2\)};
	\diagram*{(u1)--[line width=1pt,double distance=1pt](u2)--[line 
		width=1pt,double 
		distance=1pt](u3),u2--[photon](c1)--[photon](l2),(c1)--[photon](l3),(l1)--[line
		width=1pt,double distance=1pt](l2)--[line width=1pt,double 
		distance=1pt](l3)--[line width=1pt,double distance=1pt](l4)};
	\end{feynman}
	\end{tikzpicture}+\begin{tikzpicture}[baseline=(c1)]
	\begin{feynman}
	\vertex (c1) at (0,0);
	\vertex (u1) at (-1.5,1);
	\vertex (u2) [dot] at (0,1){};
	\vertex (u3) at (1.5,1) {\(1\)};
	\vertex (l1) at (-1.5,-1);
	\vertex (l2) at (-0.75,-1);
	\vertex (l3) at (0.75,-1);
	\vertex (l4) at (1.5,-1) {\(2\)};
	\diagram*{(u1)--[line width=1pt,double distance=1pt](u2)--[line 
		width=1pt,double 
		distance=1pt](u3),u2--[photon](l2),u2--[photon](l3),(l1)--[line
		width=1pt,double distance=1pt](l2)--[line width=1pt,double 
		distance=1pt](l3)--[line width=1pt,double distance=1pt](l4)};
	\end{feynman}
	\end{tikzpicture}
	+\begin{tikzpicture}[baseline=(c1)]
	\begin{feynman}
	\vertex (c1) at (0,0);
	\vertex (u1) at (-1.5,1);
	\vertex (u2) [dot] at (0,1){};
	\vertex (u3) at (1.5,1) {\(1\)};
	\vertex (l1) at (-1.5,-1);
	\vertex (l2) at (-0.75,-1);
	\vertex (l3) at (0.75,-1);
	\vertex (l4) at (1.5,-1) {\(2\)};
	\diagram*{(u1)--[line width=1pt,double distance=1pt](u2)--[line 
		width=1pt,double 
		distance=1pt](u3),u2--[photon](c1)--[photon](l2),(c1)--[photon](l3),(l1)--[line
		width=1pt,double distance=1pt](l2)--[line width=1pt,double 
		distance=1pt](l3)--[line width=1pt,double distance=1pt](l4)};
	\end{feynman}
	\end{tikzpicture}\nonumber\\
	&+\begin{tikzpicture}[baseline=(c1)]
	\begin{feynman}
	\vertex (c1) at (0,0);
	\vertex (u1) at (-1.5,-1);
	\vertex (u2) at (0,-1);
	\vertex (u3) at (1.5,-1) {\(2\)};
	\vertex (l1) at (-1.5,1);
	\vertex (l2) [dot] at (-0.75,1){};
	\vertex (l3) at (0.75,1);
	\vertex (l4) at (1.5,1) {\(1\)};
	\diagram*{(u1)--[line width=1pt,double distance=1pt](u2)--[line 
		width=1pt,double 
		distance=1pt](u3),u2--[photon](l2),u2--[photon](l3),(l1)--[line
		width=1pt,double distance=1pt](l2)--[line width=1pt,double 
		distance=1pt](l3)--[line width=1pt,double distance=1pt](l4)};
	\end{feynman}
	\end{tikzpicture}
	+\begin{tikzpicture}[baseline=(c1)]
	\begin{feynman}
	\vertex (c1) at (0,0);
	\vertex (u1) at (-1.5,-1);
	\vertex (u2) at (0,-1);
	\vertex (u3) at (1.5,-1) {\(2\)};
	\vertex (l1) at (-1.5,1);
	\vertex (l2) [dot] at (-0.75,1) {};
	\vertex (l3) at (0.75,1);
	\vertex (l4) at (1.5,1) {\(1\)};
	\diagram*{(u1)--[line width=1pt,double distance=1pt](u2)--[line 
		width=1pt,double 
		distance=1pt](u3),u2--[photon](c1)--[photon](l2),(c1)--[photon](l3),(l1)--[line
		width=1pt,double distance=1pt](l2)--[line width=1pt,double 
		distance=1pt](l3)--[line width=1pt,double distance=1pt](l4)};
	\end{feynman}
	\end{tikzpicture}
	\end{align*}
	\caption{Diagrams contributing to the 2PM eikonal phase. Dots represent 
	spin vertices. Mirrored diagrams 
	are 
		omitted.}
	\label{fig:2pm}
\end{figure}

At  2PM there are more diagrams involved, as shown in \cref{fig:2pm}. 
The final result for the scalar contribution reads
\begin{equation}\label{2pmscal}
i\chi_{\text{NE},0}=\frac{i\kappa^4 m_1m_2(m_1+m_2) 
(15\gamma^2-3)}{8\cdot32^2\pi 
	(-b^2)\sqrt{\gamma^2-1}},
\end{equation}
while the term linear in spin is
\begin{equation}\label{2pm}
i\chi_{\text{NE},1}=\frac{i\kappa^4 m_1m_2 \gamma (5\gamma^2-3) }{4\cdot32^2\pi 
	(-b^2)^{\frac{3}{2}}(\gamma^2-1)^{\frac{3}{2}}}\left(\left(4m_1+3m_2\right)(b\cdot S_1\cdot u_2)-\left(4m_2+3m_1\right)(b\cdot S_2\cdot u_1)\right)~.
\end{equation}
 We find again complete agreement with the 
literature (see e.g.
\cite{DiVecchia:2020ymx, Damour:2020tta, Aoude:2020ygw, Bern:2021dqo, Jakobsen:2021zvh}).

\subsection{A brief comparison to WQFT and HEFT approaches}
 We now make a comparison of the GWL approach with similar methods that have 
 been recently proposed in the literature. 
 In particular, the analogy between this work and 
 the 
 WQFT method of \cite{Jakobsen:2021zvh} 
  is evident. They both start from the same supersymmetric
  worldline model and the
classical limit is performed at the Lagrangian level. The main difference is 
that
the WQFT approach treats worldline fields and the bulk fields (i.e. the 
graviton) on the same footing, whereas in our approach all 
worldline integrals are solved first, leading to the representation
of the eikonal phase as the VEV of GWLs. What seems at first a mere 
prescription for the order of integrations (leading to a re-organization of 
diagrams and Feynman rules), has in fact deeper consequences. 
By solving all worldline integral first, we are
able to identify how Wilson-line operators must be modified 
in order to capture the classical dynamics of spinning objects interacting 
gravitationally at higher orders in the PM expansion. This 
opens up the 
possibility to exploit techniques that have 
been developed over 
the years for the renormalization of Wilson loops and their generalizations
\cite{Korchemskaya:1994qp, Korchemskaya:1996je, Drummond:2008vq, 
Caron-Huot:2010ryg, Becher:2010pd, 
Moult:2018jjd, Beneke:2018gvs, Beneke:2024cpq}. 

Moreover, there is another interesting feature of the GWL approach, which 
follows from the crucial role played by the soft expansion. By 
taking the route over the soft GWL we have been able to 
identify the correct spin tensor in terms of the initial data of the Grassmann 
fields $\psi_a$ without ever comparing to the classical EFT. Therefore, the GWL 
approach brings the 
advantage that one does not need knowledge of the 
correct EFT of gravity in order to identity the classical degrees of freedom.

We can also comment on the connection with a conceptually related 
approach 
\cite{Brandhuber:2021eyq,Brandhuber:2023hhy,Brandhuber:2023hhl}, which uses a 
Heavy Mass Effective Theory (HEFT) to construct the classical gravitational 
scattering 
amplitude. This approach builds on the fact that in the classical limit the 
masses of the matter particles are much larger than the momentum transfer. 
In this limit one can define the so called HEFT amplitudes 
$\mathcal{A}_n(q,\bar{p})$, which are the classical $\mathcal{O}(\bar{m}^2)$ 
pieces of the two-massive $n-2$-graviton amplitude. Here 
$\bar{p}=\frac{1}{2}(p_1+p_1')$ is the average\footnote{Using the more conventional $\hbar\rightarrow 0$ limit for scattering 
amplitudes one would run into the problem of having feed down terms since 
$q^2=\pm 2pq$ does not have a homogeneous scaling in $\hbar$, which is avoided with the barred variables.} of incoming and outgoing 
momentum and $\bar{m}^2=\bar{p}^2$. These on-shell amplitudes are then used to 
construct loop integrands via unitarity methods. 

In spite of this computational 
difference in constructing the amplitudes there is a clear connection between 
the two approaches, which originates in the way the limit is implemented to extract 
the classical pieces of the amplitudes. The clear advantage of the HEFT 
approach lies in the fact that by expanding in the masses and choosing the 
appropriate dynamical variables one can show that only two massive particle 
irreducible (2mPI) cut diagrams contribute to the classical amplitude, while  
all other diagrams are either quantum or iteration pieces. 
The use of dressed propagators in the GWL approach exponentiates the amplitude, hence taking care of all iteration pieces 
automatically. Furthermore, the use of time symmetric worldline correlators in the GWL approach leads to the identification of the background parameters $p_1$ and $p_2$ of the particle trajectories with the average of the particles momenta in the far past and far future. We therefore expect the on-shell HEFT amplitudes to directly 
correspond to the on-shell version of the classical GWL.

To better appreciate the close relation between the two methods, we
consider the following simple calculation.
Contracting the classical one-graviton vertex (we do not distinguish between 
spinning and non-spinning vertices here) with the polarization tensor 
$\varepsilon_{\mu\nu}=\varepsilon_\mu\varepsilon_\nu$, we 
obtain  
\begin{equation}
\varepsilon_\mu\varepsilon_\nu \Big(\begin{tikzpicture}[baseline=(a)]
\begin{feynman}
\vertex (a);
\vertex [right=of a] (b) ;
\vertex [above=of b] (c) {\(k,\mu\nu\)};
\vertex [right=of b] (d) {\(p,S\)};
\diagram*{(a)--[line width=1pt,double distance=1pt] (b)--[line width=1pt,double 
	distance=1pt](d), (b)--[photon](c)};
\end{feynman}
\end{tikzpicture}\Big)=\frac{1}{2}\int_k 
\hat{\delta}(pk)\mathcal{A}_3^{\text{HEFT}}(p)~,
\end{equation}
where the diagram on the l.h.s. represents the corresponding vertex in the GWL. 
A similar analogy holds for the HEFT Compton amplitude, which 
reads
\begin{equation}
\varepsilon_\mu\varepsilon_\nu\varepsilon_\varrho
\varepsilon_\sigma\Big(\begin{tikzpicture}[baseline=(c1)]
\begin{feynman}
\vertex (c1) at (0,0);
\vertex (u1) at (-1.5,-1);
\vertex (u2) at (0,-1);
\vertex (u3) at (1.5,-1) {\(p,S\)};
\vertex (l1) at (-1.5,1);
\vertex (l2) at (-0.75,1) {\(k_1,\mu\nu\)};
\vertex (l3) at (0.75,1){\(k_2,\varrho\sigma\)};
\vertex (l4) at (1.5,1) ;
\diagram*{(u1)--[line width=1pt,double distance=1pt](u2)--[line 
	width=1pt,double distance=1pt](u3),(u2)--[photon](l2),(u2)--[photon](l3)};
\end{feynman}
\end{tikzpicture}+\begin{tikzpicture}[baseline=(c1)]
\begin{feynman}
\vertex (c1) at (0,0);
\vertex (u1) at (-1.5,-1);
\vertex (u2) at (0,-1);
\vertex (u3) at (1.5,-1) {\(p,S\)};
\vertex (l1) at (-1.5,1);
\vertex (l2) at (-0.75,1) {\(k_1,\mu\nu\)};
\vertex (l3) at (0.75,1){\(k_2,\varrho\sigma\)};
\vertex (l4) at (1.5,1) ;
\diagram*{(u1)--[line width=1pt,double distance=1pt](u2)--[line 
	width=1pt,double 
	distance=1pt](u3),u2--[photon](c1)--[photon](l2),(c1)--[photon](l3)};
\end{feynman}
\end{tikzpicture}\Big)=\frac{1}{2}\int_{k_1,k_2} 
\hat{\delta}(p(k_1+k_2))\mathcal{A}_4^{\text{HEFT}}(p)~.
\end{equation}
Note that the relations above always include an additional factor of $1/2$ and 
the 
orthogonality delta function. This is simply due to the fact that when the GWL 
is 
defined in position space it directly gives the impact parameter space 
amplitude. A vertex from the upper and lower line then together yields the 
measure 
$\frac{1}{4}\int_q\hat{\delta}(qp_1)\hat{\delta}(qp_2)e^{-iqb}=\frac{1}{4m_1m_2\sqrt{\gamma^2-1}}\int_q^{D-2}e^{i\vec{q}\vec{b}}$.

Finally, we observe that this connection between the HEFT and the GWL methods 
sheds light on the relation between the HEFT and the WQFT approach. Indeed, by 
mapping the on-shell amplitudes of HEFT onto the vertices of the GWL, we are in 
fact mapping the worldline vertices of the WQFT to the operators of the HEFT. 
In this regard, the GWL method can be considered as somewhat intermediate, since the connection between the different approaches becomes only apparent by solving for the worldline fields and bulk field separately. It is a 
supersymmetric worldline approach like the WQFT, where the classical limit is 
performed at the level of an 
intrinsically off-shell object such as the Lagrangian, unlike the on-shell 
amplitude framework. However, it is also formulated in terms of operators in a 
mode-expansion 
(specifically, the soft expansion) in analogy with the heavy-mass expansion of 
HEFT.

\section{Conclusions}
\label{sec_concl}

In this work we have 
 provided
two representations for Generalized Wilson 
lines operators that are relevant for gravitational scattering amplitudes
at the classical and at the quantum level, respectively. 
Specifically, we have derived a representation for a soft GWL for spin $1/2$ particles
(\cref{eq:GWL_spin}) and a classical GWL (\cref{clGWL} and \cref{clGWLmom}) which generates observables with linear corrections
in spin. In this way we have extended results for the scalar case presented in \cite{Bonocore:2020xuj}.

The soft GWL is a compact tool to describe the factorization and the 
exponentiation 
of
subleading soft gravitons dressing the external states of a quantum scattering 
amplitude. As in the scalar case, up to the order in perturbation theory 
implemented  
in this work (i.e. next-to-eikonal), soft GWLs exponentiate not only the 
single-graviton factor of soft theorems but also pairwise correlations of soft 
emissions. The final result shows a remarkable simplicity, where the only
modification w.r.t. the scalar case is given by the single-graviton soft 
factor depending on the Lorenz spin generator. No modification is 
necessary for the correlated emissions, which thus remain spin-independent.

The classical GWL is constructed in analogy with the soft one, with two main 
differences: a classical limit that is performed at the Lagrangian 
level and different boundary conditions. The boundary conditions in particular 
have 
been chosen according to an interrupted straight line from $-\infty$ to 
$+\infty$. Classical spin variables are mapped to quantum spin 
numbers which tend to infinity in the classical limit and are thus enhanced compared to the soft limit. 
This different power 
counting implies an important difference between the soft and the classical 
GWL: both single and double graviton emissions depend on spin in the classical case, in agreement 
with other methods explored in the literature. We have cross-checked our 
representation by rederiving the eikonal phase at 2PM.

In spite of these differences, it is actually the analogy between the soft and 
the classical GWL that turns out to be particularly useful. Both of them are derived 
in the supersymmetric worldline formalism by assuming a soft 
expansion for the background field. In this way, the correspondence 
between the soft expansion and the PM expansion provides a map between quantum 
spin (represented by the worldline Grassmann variables) and the classical spin, 
thus not requiring any reference to the worldline EFT. 

We have also discussed the relation with similar methods in the literature, and 
specifically with the WQFT and the HEFT. The relation with former is 
straightforward, since GWLs are constructed by integrating out the same 
worldline degrees of freedom that appear in the WQFT. We have also shown a 
clear connection between the GWL and the HEFT, by matching the single and double graviton emissions of the GWL with the on-shell HEFT amplitudes. We therefore conclude that the 
classical GWL provides a bridge between the WQFT and the HEFT, making their 
equivalence 
manifest.  

On a more technical side, some new features have been analyzed in the derivation 
of the GWL from first principles in the supersymmetric worldline formalism.
Specifically, we have thoroughly discussed 
the role of fermion doubling in the construction of the GWL and how to consistently integrate out the auxiliary 
Grassmann fields.
Additionally, in analogy with the gauge theory case discussed in 
\cite{Bonocore:2020xuj}, we have showed that the open line topology of the dressed 
propagator with twisted boundary conditions enables to evaluate the 
supersymmetry charge $Q$ at asymptotic time $t_f\rightarrow \infty$, such that 
it reduces to the numerator contribution of the free Dirac propagator.

The results of this work can be extended in different directions. For instance,  at the 
quantum level a systematic treatment of collinear interactions, 
which become relevant at subleading powers as shown in \cite{Beneke:2022ehj}, would be desirable. At the classical level, a very interesting direction 
is the inclusion
of higher spin particles in the GWL framework \cite{newpaper}.

\section*{Acknowledgments}
This work was supported by the Deutsche Forschungsgemeinschaft (DFG) through 
the 
Research Training Group "GRK 2149: Strong and Weak Interactions – from Hadrons 
to Dark Matter" and by the Excellence Cluster ORIGINS funded by the DFG under 
Grant No. EXC- 2094-390783311. AK gratefully acknowledges the support and the hospitality of the CERN Theoretical Physics Department during early stages of this work.

\appendix

\section{Fermionic path integral}
\label{sec:ferm}
In order to integrate out 
$\psi_2$ from the fermionic part of the path integral, 
we have to deal with the dependence on $\psi_2$ in the boundary term 
$\bar{\chi}_N\chi_{N-1}$ and in the kinetic term to show that indeed $\psi_2$ 
can be integrated out. In doing so we derive the propagator for the field 
$\psi_1$.

We start by taking the discretised expression for the boundary and kinetic term 
\begin{equation} 
	\prod_{j=1}^{N-1}\left(d^4\bar{\chi}_jd^4\chi_{j}\right)\exp\left\{\bar{\chi}_N\chi_{N-1}-\sum_{n=1}^{N-1}\bar{\chi}_n(\chi_{n}-\chi_{n-1})\right\},
\end{equation}
where we then expand $\chi$ into real and imaginary parts. Here we need to keep in mind that the sum contains $\chi_0=\chi_i$, which is our initial condition for $\chi$. We will for this calculation keep the boundaries $\bar{\chi}_f$ and $\chi_i$ fixed and explicit. Additionally, with the expansion into real and imaginary parts we need to change the integration measure. Luckily, due to the $\sqrt{2}$ factors, the Jacobian determinant of this transformation has modulus one, thus 
\begin{equation} 
	\prod_{j=1}^{N-1}d^4\bar{\chi}_jd^4\chi_{j}=	\prod_{j=1}^{N-1}d^4\psi_{1,j}d^4\psi_{2,j}.
\end{equation}
All together this yields
\begin{align}
	&\nonumber\int \prod_{j=1}^{N-1}\left(d^4\psi_{1,j}d^4\psi_{2,j}\right) \exp\Bigg\{-\frac{1}{2}\sum_{n=2}^{N-1}(\psi_{1,n}-i\psi_{2,n})(\psi_{1,n}-\psi_{1,n-1}+i\psi_{2,n}-i\psi_{2,n-1})+\\
	&+\frac{1}{\sqrt{2}}\bar{\chi}_f(\psi_{1,N-1}+i\psi_{2,N-1})-\frac{1}{\sqrt{2}}\chi_i(\psi_{1,1}-i\psi_{1,1})-\frac{1}{2}(\psi_{1,1}-i\psi_{1,1})(\psi_{1,1}+i\psi_{1,1})\Bigg\}.
\end{align}
Now let us rewrite the exponential in matrix notation to make the structure more apparent. 
With the following definitions
\begin{align}
	M&=\frac{1}{2}\begin{pmatrix}
		0&1&0&0&\dots\\
		-1&0&1&0&\dots\\
		0&-1&0&1&\dots\\
		0&0&-1&0&\dots\\
		\vdots&\vdots&\vdots&\vdots&\ddots\\
	\end{pmatrix}\\
	\vec{\psi}_{2,B}&=\begin{pmatrix}
		\chi_i\\
		0\\
		\vdots\\
		0\\
		\bar{\chi}_f
	\end{pmatrix}\\
	\vec{\psi}_{\psi_1}&=\begin{pmatrix}
		\psi_{1,2}-2\psi_{1,1}~~~~~\\
		\psi_{1,3}-2\psi_{1,2}+\psi_{1,1}\\
		\vdots\\
		\psi_{1,N-1}-2\psi_{1,N-2}+\psi_{1,N-3}\\
		~~~~~-2\psi_{1,N-1}+\psi_{1,N-2}
	\end{pmatrix}\\
	\vec{\psi}_{1,B}&=\begin{pmatrix}
		-\chi_i\\
		0\\
		\vdots\\
		0\\
		\bar{\chi}_f
	\end{pmatrix}\\
	\vec{\psi}_1&=\begin{pmatrix}
		\psi_{1,1}\\\psi_{1,2}\\\vdots\\\psi_{1,N-1}
	\end{pmatrix} ~~~~\text{~and~}~~~~~\vec{\psi}_2=\begin{pmatrix}
		\psi_{2,1}\\\psi_{2,2}\\\vdots\\\psi_{2,N-1}
	\end{pmatrix}
\end{align}
the exponential takes the form 
\begin{align}
	\exp\left\{	-\frac{1}{2}\vec{\psi}_1^{~T}M\vec{\psi}_1+\frac{1}{\sqrt{2}}\vec{\psi}_{1,B}^{~T}\vec{\psi}_1-\frac{1}{2}\vec{\psi}_2^{~T}M\vec{\psi}_2+i\left(\frac{1}{\sqrt{2}}\vec{\psi}_{2,B}+\frac{1}{2}\vec{\psi}_{\psi_1}\right)^{T}\vec{\psi}_2\right\}.
\end{align}
We can integrate out the field $\psi_2$ using the standard Gaussian integral for Grassmann numbers yielding
\begin{align}
	\exp\left\{-\frac{1}{2}\vec{\psi}_1^{~T}M\vec{\psi}_1\!+\!\frac{1}{\sqrt{2}}\vec{\psi}_{1,B}^{~T}\vec{\psi}_1\!+\!\frac{1}{2}\left(\frac{1}{\sqrt{2}}\vec{\psi}_{2,B}\!+\!\frac{1}{2}\vec{\psi}_{\psi_1}\right)^{T}M^{-1}\left(\frac{1}{\sqrt{2}}\vec{\psi}_{2,B}\!+\!\frac{1}{2}\vec{\psi}_{\psi_1}\right)\right\}.
\end{align}
The inverse of $M$ is given by 
\begin{equation}
	M^{-1}=2\begin{pmatrix}
		0&-1&0&-1&0&-1&\dots\\
		1&0&0&0&0&0&\dots\\
		0&0&0&-1&0&-1&\dots\\
		1&0&1&0&0&0&\dots\\
		0&0&0&0&0&-1&\dots\\
		1&0&1&0&1&0&\dots\\
		\vdots&\vdots&\vdots&\vdots&\vdots&\vdots&\ddots\\		
	\end{pmatrix}.
\end{equation}
We have thus successfully\footnote{Actually the integration also creates the prefactor $\text{Pf}(M)=\frac{1}{2^{N-1}}$, which is cancelled by the integration over $\psi_{1}$ as this gives $\text{Pf}(G)=2^{N-1}$.} integrated out $\psi_2$, but apparently we did not simplify things. On the contrary, at first glance we seem to have made things a lot more complicated. Luckily, there is still a lot of simplification that can be done. The vectors $\psi_{\psi_1}$ still depend on $\psi_1$, and they do so through this matrix equation
\begin{equation} 
	\vec{\psi}_{\psi_1}=A\vec{\psi}_1~~~~\text{with}~~~~A=\begin{pmatrix}
		-2&1&0&0&\dots\\
		1&-2&1&0&\dots\\
		0&1&-2&1&\dots\\
		0&0&1&-2&\dots\\
		\vdots&\vdots&\vdots&\vdots&\ddots\\
	\end{pmatrix}.
\end{equation}
We can also connect the boundary vectors by 
\begin{equation} 
	\vec{\psi}_{1,B}=C\vec{\psi}_{2,B}~~~\text{with}~~~C=\begin{pmatrix}
		-1&0\\
		0&\mathds{1}\\
	\end{pmatrix}
\end{equation}
and a quick calculation gives
\begin{equation} 
	\frac{1}{4}\vec{\psi}_{2,B}^{~T}M^{-1}\vec{\psi}_{2,B}=\bar{\chi}_f\chi_i.
\end{equation}
Thus, the exponential can be rewritten as
\begin{align}
	\exp\left\{\bar{\chi}_f\chi_i-\frac{1}{2}\vec{\psi}_1\underbrace{(M-\frac{1}{4}AM^{-1}A)}_{G}\vec{\psi}_1+\frac{1}{\sqrt{2}}\left((C-\frac{1}{2}AM^{-1})\vec{\psi}_{2,B}\right)^T\vec{\psi}_1\right\},
\end{align}
Now we see the desired structure emerge. The matrix $G$ defines the propagator 
for $\psi_1$ through its inverse and the boundary term $\bar{\chi}_f\chi_i$ is 
precisely the one we expect to appear since 
$\braket{\bar{\chi}_f|\chi_i}=e^{\bar{\chi}_f\chi_i}$. 
All that is left to do is to expand the field $\psi_1$ into a background field 
and a quantum fluctuation, add a source term and perform the integration over 
$\psi_1$. The background-quantum split is given by 
\begin{equation} 
	\psi_{1,n}=\frac{\eta}{2}+\tilde{\psi}_{1,n},
\end{equation}
where the quantum field $\tilde{\psi}_{1,n}$ vanishes at the boundary for $n=0$ and $n=N$. The source term takes the form $\vec{J}^{~T}\vec{\tilde{\psi}}_1$. Integrating over the remaining Grassmann fields yields
\begin{equation}\label{exponential}
	\exp\left\{\bar{\chi}_f\chi_i-\frac{1}{2}\vec{J}^{~T}G^{-1}\vec{J}\right\}=\exp\left\{\bar{\chi}_f\chi_i-\frac{1}{2}J_iG^{-1}_{ij}J_j\right\}.
\end{equation}
The form of $G^{-1}$ is 
\begin{equation} 
	G^{-1}_{ij}=\left\{\begin{array}{ll}
		-\frac{1}{2}&\text{for~}i<j\\
		0&\text{for~}i=j\\
		\frac{1}{2}&\text{for~}i>j
	\end{array}\right. ~.
\end{equation}
Here we also needed the fact that 
\begin{equation} 
	\left((C-\frac{1}{2}AM^{-1})\vec{\psi}_{2,B}\right)^TG^{-1}\left((C-\frac{1}{2}AM^{-1})\vec{\psi}_{2,B}\right)=0.
\end{equation}
Now $G^{-1}_{ij}$ is precisely the propagator 
\begin{equation} 
	\contraction{}{\psi}{_{1,i}^a}{\psi}\psi_{1,i}^a\psi_{1,j}^b=-\eta^{ab}G^{-1}_{ij},
\end{equation}
thus the continuum version is 
\begin{equation}\label{wl_fermionicpathintegral_psipropagator}
	\contraction{}{\psi}{_{1}^a(t)}{\psi}\psi_{1}^a(t)\psi_{1}^b(t')
	=\frac{1}{2}\eta^{ab}(\theta(t-t')-\theta(t'-t))~,
\end{equation}
as quoted in  \cref{fermionProp} in the main text. 

Equipped with this propagator, we can elaborate \cref{exponential} to get the 
final representation for the fermionic path integral.  
Firstly, 
we see that the boundary term $\bar{\chi}_f\chi_i$ is entirely generated 
through the $\psi_2$ integration and not problematic since it is fixed. Secondly, observe that the continuum propagator for 
$\psi_1(t)$ is 
precisely the antisymmetrized Green's function for the differential operator 
$\frac{d}{dt}$. 
It is therefore justified to symbolically write the path integral over the real 
Grassmann field $\psi_1(t)$ in the usual path integral notation 
\begin{equation}\label{eq:app_realgrassmannpathintegral}
\bra{\bar{\chi}_f}e^{\hat{A}T}\ket{\chi_i}=\int \mathcal{D}\psi_1 \exp \left\{ 
\bar{\chi}_f\chi_i+\int_0^T dt 
\left(-\frac{1}{2}\psi_1\dot{\psi}_1+A_{\text{px}}(\psi_1)\right) \right\},
\end{equation}
in agreement with \cref{fermionPath}. 

\section{A Theorem about fermionic path integrals}\label{app:theo}

We prove a theorem about expectation values of Majorana operators that 
bears some similarity to Wick's theorem. 
\begin{theorem}
	For a product of Majorana fermion operators $\hat{\psi}_1^a$ the 
	expectation value is given by
	$$\braket{\prod_{i=1}^{n}\hat{\psi}_1^{a_i}}=\sum\limits_{\begin{subarray}
		{c}\text{possible}\\\text{contractions}
		\end{subarray}}\left[\text{contractions}\times\prod\limits_{i\text{ not 
			contracted}}\braket{\hat{\psi}_1^{a_i}}\right].$$
\end{theorem}
\begin{proof}
	Let us check the $n=2$ case. Expanding everything out and writing it in 
	$\psi^\dag\psi$-ordering gives
	$$\braket{\hat{\psi}_1^a\hat{\psi}^b_1}=\frac{1}{2}\braket{\hat{\psi}^a\hat{\psi}^b+\hat{\psi}^{a\dagger}\hat{\psi}^b+\hat{\psi}^{a\dagger}\hat{\psi}^{b\dagger}-\hat{\psi}^{b\dagger}\hat{\psi}^{a\dagger}+\eta^{ab}}.$$
	Given the anticommuting nature of Grassmann numbers we can evaluate these 
	operators directly on the coherent states and rewrite this as 
	$$\frac{\eta^{ab}}{2}+\braket{\hat{\psi}_1^a}\braket{\hat{\psi}_1^b}.$$
	Now let us assume that the theorem holds for $n-1$. Then the statement for 
	$n$ follows from examining 
	\begin{align*}
	\Big\langle\left(\prod_{i=1}^{n-1}\hat{\psi}_1^{a_i}\right)\frac{\hat{\psi}^{a_n}+\hat{\psi}^{a_n\dagger}}{\sqrt{2}}\Big\rangle.
	\end{align*}
	The term with $\hat{\psi}^{a_n}$ is already in the correct order, but we 
	need to permute $\hat{\psi}^{a_n\dagger}$ trough all $\hat{\psi}_1^{a_i}$ 
	to the very left. From the anticommutation relation we can show that 
	$$\hat{\psi}_1^a 
	\hat{\psi}^{b\dagger}=-\hat{\psi}^{b\dagger}\hat{\psi}_1^a+\frac{\delta^{ab}}{\sqrt{2}},$$
	and we pick up a contraction with all $n-1$ Majorana fermions and an 
	$(-1)^{n-1}$ in front of the permuted term. Since it is then in the correct 
	order, we can evaluate it on the bra state and permute the Grassmann number 
	back to the right, cancelling the $(-1)^{n-1}$. So all in all we have 
	achieved 
	\begin{align*}
	\braket{\left(\prod_{i=1}^{n-1}\hat{\psi}_1^{a_i}\right)\hat{\psi}_1^{a_n}}=\braket{\prod_{i=1}^{n-1}\hat{\psi}_1^{a_i}}\braket{\hat{\psi}_1^{a_n}}+\sum_{j=1}^{n-1}
	(-1)^{n-j-1}\contraction{}{\hat{\psi}}{_1^{a_j} }{\hat{\psi}}
	\hat{\psi}_1^{a_j}  \hat{\psi}_1^{a_n} \braket{\prod_{\begin{subarray}
			{c} i=1\\i\neq j
			\end{subarray}}^{n}\hat{\psi}_1^{a_i}}.
	\end{align*}
	On the remaining expectation values we can use the theorem for ($n-1$) to 
	get that the desired statement for $n$. The first term gives all the terms 
	where $\hat{\psi}_1^{a_n}$ is not contracted. The second part gives all 
	possibilities of contractions including the $n$th operator 
	$\hat{\psi}_1^{a_n}$. 
\end{proof}

\section{Eikonal phase at 2PM}
\subsection{Feynman-Rules}
\label{sec:feynman}
The GWL sources the 
graviton field and we can divide the resulting types of vertices into powers of 
the gravitational coupling and the spin tensor. Let us first discuss the 
$\mathcal{O}(\kappa)$ vertices. It will be helpful in the ensuing calculations 
to divide diagrams into different topologies depending on their scaling with 
the masses of the different particles. Therefore we will introduce the velocity 
$p_i^\mu=m_iu_i^\mu$ and also rescale the Grassmann fields as 
$\psi_i^a\rightarrow \sqrt{m_i}\psi_i$. Additionally, we introduce some helpful 
shorthand notation
\begin{equation} 
\int \frac{d^D k}{(2\pi)^D}=\int_k
\end{equation}
and \begin{equation} 
(2\pi)^D\delta^{(D)}(x)=\hat{\delta}(x).
\end{equation}
The single graviton vertices then become in momentum space
\begin{equation} 
\begin{tikzpicture}[baseline=(a)]
\begin{feynman}
\vertex (a);
\vertex [right=of a] (b) ;
\vertex [above=of b] (c) {\(k,\mu\nu\)};
\vertex [right=of b] (d) {\(u_i\)};
\diagram*{(a)--[line width=1pt,double distance=1pt] (b)--[line width=1pt,double 
distance=1pt](d), (b)--[photon](c)};
\end{feynman}
\end{tikzpicture}	=-i\frac{\kappa}{2}m_i u_i^\mu u_i^\nu 
\int_k\hat{\delta}(u_ik),
\end{equation} 
and 
\begin{equation} 
\begin{tikzpicture}[baseline=(a)]
\begin{feynman}
\vertex (a);
\vertex [right=of a, dot] (b) {};
\vertex [above=of b] (c) {\(k,\mu\nu\)};
\vertex [right=of b] (d) {\(u_i,S_i\)};
\diagram*{(a)--[line width=1pt,double distance=1pt] (b)--[line width=1pt,double 
distance=1pt](d), (b)--[photon](c)};
\end{feynman}
\end{tikzpicture}	=\frac{\kappa}{2}m_iu_i^\mu 
S_i^{\nu\varrho}\int_k\hat{\delta}(u_ik)k_\varrho.
\end{equation} 
Also at $\mathcal{O}(\kappa)$ we have the three graviton vertex. We adopt the 
form also used in \cite{Kalin:2020mvi} 
\begin{align}
&\nonumber\begin{tikzpicture}[baseline=(b)]
\begin{feynman}
\vertex (b) ;
\vertex [below left=1cm of b] (a) {\(k_1,ab\)};
\vertex [above=1cm of b] (c){\(k_3,ef\)};
\vertex [below right=1cm of b] (d){\(k_2,cd\)};
\diagram*{(a)--[photon] (b)--[photon](d),(b)--[photon](c)};
\end{feynman}
\end{tikzpicture}=\frac{i\kappa}{4}\Big[\nonumber 
4k_1k_2(\eta^{af}\eta^{bd}\eta^{ce}+\eta^{ae}\eta^{bd}\eta^{cf})\\&\nonumber+4k_2k_3(\eta^{ae}\eta^{bc}\eta^{df}+\eta^{ac}\eta^{be}\eta^{df})
 +4k_1k_3(\eta^{ad}\eta^{bf}\eta^{ce}+\eta^{ac}\eta^{bf}\eta^{de})\\&\nonumber 
+(k_1k_2+k_2k_3+k_1k_3)(\eta^{ab}\eta^{cd}\eta^{ef}-2\eta^{ae}\eta^{bf}\eta^{cd}
-2\eta^{ab}\eta^{ce}\eta^{df}-2\eta^{ac}\eta^{bd}\eta^{ef})\\
&\nonumber 
-4\Big(\eta^{ad}\eta^{ce}k_1^{f}k_2^{b}
+\eta^{ae}\eta^{bc}k_1^{d}k_2^{f}
+\eta^{ac}\eta^{de}k_2^{f}k_3^{b}
+\eta^{ae}\eta^{cf}k_2^{b}k_3^{d} 
+\eta^{af}\eta^{ce}k_1^{d}k_3^{b}
+\eta^{ac}\eta^{be}k_1^{f}k_3^{d}\Big)\\
&+2\eta^{ac}\eta^{bd}(k_1^{e}k_2^{f}+k_1^{f}k_2^{e})
+2\eta^{ce}\eta^{df}(k_2^{a}k_3^{b}
+k_2^{b}k_3^{a})+2\eta^{ae}\eta^{bf}(k_1^{c}k_3^{d}
+k_1^{d}k_3^{c})
 \Big].
\end{align}
Let us turn to the $\mathcal{O}(\kappa^2)$ vertices. In momentum space they can 
be written as 
\begin{align}
&\begin{tikzpicture}[baseline=(b)]
\begin{feynman}
\vertex[left=of b] (a);
\vertex (b);
\vertex[right=of b] (c) {\(u_i\)};
\vertex[above left= 1cm of b] (d) {\(k,\mu\nu\)};
\vertex[above right= 1cm of b] (e) {\(l,\varrho\sigma\)};
\diagram*{(a)--[line width=1pt,double distance=1pt](b)--[line width=1pt,double 
distance=1pt](c),(d)--[photon](b)--[photon](e)};
\end{feynman}
\end{tikzpicture}=\frac{m_i\kappa^2}{4}\!\int_{k,l}\!\hat{\delta}(u_i(k+l))\left[klu_i^\mu
 u_i^\nu u_i^\varrho u_i^\sigma J_1-u_i^\nu u_i^\varrho u_i^\sigma l^\mu 
J_2+4i\eta^{\mu\varrho}u_i^\nu u_i^\sigma \right]
\end{align}
and 
\begin{align}
&\begin{tikzpicture}[baseline=(b)]\nonumber
\begin{feynman}
\vertex[left=of b] (a);
\vertex [dot] (b) {};
\vertex[right=of b] (c) {\(u_i,S_i\)};
\vertex[above left= of b] (d) {\(k,\mu\nu\)};
\vertex[above right= of b] (e) {\(l,\varrho\sigma\)};
\diagram*{(a)--[line width=1pt,double distance=1pt](b)--[line width=1pt,double 
distance=1pt](c),(d)--[photon](b)--[photon](e)};
\end{feynman}
\end{tikzpicture}\\=&\nonumber-\frac{i\kappa^2m_i}{4}S_i^{ab}\int_{k,l}\hat{\delta}((k+l)u_i)\Big[2J_2(u_i^\varrho
 u_i^\sigma k^b l^\mu \eta^{\nu a}+2 u_i^\mu u_i^\varrho l^b l^\nu \eta^{\sigma 
a}+2u_i^\mu u_i^\varrho \eta^{\nu[a}k^{d]}\eta^{\sigma 
[b}l^{d]})\\&+2J_1u_i^\mu u_i^\varrho u_i^\sigma \eta^{\nu a}k^b 
kl+2i\left(u_i^\varrho \eta^{\nu a}l^\mu \eta^{\sigma b}+u_i^\varrho l^a 
\eta^{\nu\sigma} \eta^{\mu b}+\frac{1}{2}pl \eta^{\mu\varrho}\eta^{\nu 
b}\eta^{\sigma a}-2u_i^\nu l^b \eta^{\mu\varrho}\eta^{\sigma a}\right)	\Big]
\end{align}
The functions $J_i$ are defined as follows 
\begin{align} 
J_1&=-\frac{i}{2}\left(\frac{1}{(lu_i+i\epsilon)^2}
+\frac{1}{(u_il-i\epsilon)^2}\right),\\
J_2&=\frac{i}{2}\left(\frac{1}{(lu_i+i\epsilon)}
+\frac{1}{(u_il-i\epsilon)}\right).	
\end{align}
One should also remember that the Feynman rules listed above must be 
supplemented with  appropriate symmetrization 
factors, when necessary.
We also observe that in the VEV of \cref{regge}, two Wilson lines are 
 shifted by 
the impact parameter $b$, which should be large in the classical limit. In 
momentum space this is done via the phase $e^{-ibk}$ for each graviton momentum 
flowing into the lower line.

\subsection{1PM eikonal phase}\label{sec:1pm}

We consider the diagrams of \cref{fig:1pm}. 
The first diagram is rather simple to evaluate. We find
\begin{equation} 
-\frac{i\kappa^2m_1m_2}{4}(\gamma^2-\frac{1}{D-2})\int_q 
\frac{\hat{\delta}(u_1q)\hat{\delta}(u_2q)}{q^2}e^{-ibq}.
\end{equation}
We then note that last diagram is obtained from the second by mirroring. This flips the 
direction of $b$ and we therefore 
obtain one diagram from the other by exchanging labels $1\leftrightarrow 
2$ and $b\rightarrow -b$.

The second diagram gives (upon writing the graviton momentum in 
the 
numerator as a derivative acting on the exponential)
\begin{equation} 
-\frac{i\kappa^2m_1m_2\gamma}{4}\frac{\partial}{\partial 
	b^\mu}S_1^{\mu\nu}u_{2\nu}\int_q 
\frac{\hat{\delta}(u_1q)\hat{\delta}(u_2q)}{q^2}e^{-ibq}.
\end{equation}
This directly tells us that the third diagram evaluates to 
\begin{equation} 
\frac{i\kappa^2m_1m_2\gamma}{4}\frac{\partial}{\partial 
	b^\mu}S_2^{\mu\nu}u_{1\nu}\int_q 
\frac{\hat{\delta}(u_1q)\hat{\delta}(u_2q)}{q^2}e^{-ibq}.
\end{equation}
The remaining Fourier transform can be performed in a suitably chosen frame. We 
take the frame where particle 1 is at rest and particle 2 moves along the 
$x$-direction. This means $u_1=(1,\vec{0})$ and 
$u_2=(\gamma,\sqrt{\gamma^2-1},\vec{0}_{D-2})$. Integrating over the delta 
constraints and performing the remaining $D-2$-dimensional Fourier transform 
yields 
\begin{equation} 
\int_q 
\frac{\hat{\delta}(u_1q)\hat{\delta}(u_2q)}{q^2}e^{-ibq}=-\frac{1}{4\pi^{(D-2)/2}}\frac{\Gamma(D/2-2)}{\sqrt{\gamma^2-1}}
|\vec{b}_\perp|^{4-D}\overset{D\rightarrow 4-2\epsilon}{\longrightarrow} 
-\frac{\Gamma(-\epsilon)}{4\pi\sqrt{\gamma^2-1}}(-b^2)^{\epsilon}.
\end{equation}
Here we have written $\vec{b}_\perp$ for the remaining spatial part of $b$ that 
is orthogonal to both velocities. Also in the last step we have written it in a 
covariant way, which assumes that we keep in mind that $b^\mu$ is taken to be 
orthogonal to the velocities. Inserting this integral and performing the 
remaining derivatives with respect to the impact parameter we find \cref{1pm}.

\subsection{2PM eikonal phase}\label{sec:2pm}
We consider the diagrams of \cref{fig:2pm}. 
The first diagram turns out the be (where we have substituted $l\rightarrow 
l-k$ and then $k\rightarrow l-k$)
\begin{align}
&\nonumber\frac{\kappa^4 
	m_2^2m_1}{32}\int_{k,l}\frac{\hat{\delta}(u_2k)\hat{\delta}(u_2l)\hat{\delta}(u_1l)}{k^2(l-k)^2}e^{ibl}\Big[J_{1,l\rightarrow
	k}
k(l-k)(\gamma^2-\frac{1}{D-2})^2+4i(\gamma 
u_2-\frac{1}{D-2}u_1)^2\\&+\frac{4i}{D-2}\left(\gamma^2-\frac{1}{D-2}\right)\Big].
\end{align}
There are a couple of things to point out here. First we can rewrite 
$k(l-k)=\frac{1}{2}(l^2-k^2-(l-k)^2)$. We note that since $l$ appears in the 
exponential it can be understood diagrammatically as the momentum exchange, 
whereas $k$ plays the role of the loop momentum. Therefore, from the rewriting 
only the first term survives because the latter two cancel a propagator pole in 
the denominator. This has the effect of leading to a $\delta(b)$ contribution 
an hence a contact interaction. What remains of the first term is the loop 
integral 
\begin{equation} 
\int_k 
\frac{\hat{\delta}(u_2k)}{k^2(k-l)^2}\left(\frac{1}{(u_1k-i\epsilon)^2}+\frac{1}{(u_1k+i\epsilon)^2}\right).
\end{equation}
We can go into the rest frame of particle 1, which leaves two poles in the 
$k^0$ component. For classical physics we need to integrate over the potential 
region. 
Although the poles lie inside the potential region, they also only occur in on half of the complex plane. Closing the contour in the other half respectively reveals that the integral does not contribute to the eikonal phase. This will 
also happen for all other diagrams involving the $J_1$ function. For the other 
terms in the first diagram we find 
\begin{equation} 
\frac{i\kappa^4m_2^2m_1}{8}\frac{D-3}{D-2}\gamma^2\int_{k,l}\frac{\hat{\delta}(u_2k)\hat{\delta}(u_2l)\hat{\delta}(u_1l)e^{ibl}}{k^2(l-k)^2}.
\end{equation}
We now turn to the second diagram and the result in $D=4$ is given by
\begin{equation} 
\frac{i\kappa^4m_2^2m_1}{32}\int_{k,l}\frac{\hat{\delta}(u_2k)\hat{\delta}(u_2l)\hat{\delta}(u_1l)e^{ibl}}{k^2l^2(l-k)^2}\left[kl(2\gamma^2-1)-l^2\gamma^2-k^2(2\gamma^2-1)-(ku_1)^2\right]
\end{equation}
There are several things to note here. First of all we can express 
$kl=\frac{1}{2}(k^2+l^2-(k-l)^2)$. As before only the $l^2$ term contributes. 
The $u_1k$ term can be evaluated 
using tensor reduction and as shown in \cite{Kalin:2020mvi} this reduces the 
$k$-integral to   
$\frac{\gamma^2-1}{8}\int_k\frac{\hat{\delta}(u_2k)}{k^2(l-k)^2}$. Therefore 
the three-point graviton diagram gives us 
\begin{equation} 
\frac{i\kappa^4m_2^2m_1}{256}(-\gamma^2-3)\int_{k,l}\frac{\hat{\delta}(u_2k)\hat{\delta}(u_2l)\hat{\delta}(u_1l)e^{ibl}}{k^2(l-k)^2}.
\end{equation}
In combination with the first diagram we find 
\begin{equation} 
\frac{i\kappa^4m_2^2m_1}{256}(15\gamma^2-3)\int_{k,l}\frac{\hat{\delta}(u_2k)\hat{\delta}(u_2l)\hat{\delta}(u_1l)e^{ibl}}{k^2(l-k)^2},
\end{equation} 
and after performing the $k$ and $l$ integration this gives us the 2PM scalar 
eikonal phase
\begin{equation} 
i\chi^{(2)}_{S^0}=\frac{i\kappa^4 m_2^2m_1(15\gamma^2-3)}{8\cdot 32^2 \pi 
	\sqrt{\gamma^2-1}\sqrt{-b^2}}.
\end{equation}

The four diagrams involving spin can be evaluated along the same lines. Here we 
need to compute the tensor integral 
\begin{equation} 
I^{\mu\nu}=\int_k \frac{\hat{\delta}(u_1k)k^\mu k^\nu}{k^2(k-l)^2}.
\end{equation}
Expanding the integral on a basis of tensors making heavy use of 
the delta constraint we find
\begin{equation} 
I^{\mu\nu}=\frac{\sqrt{-l^2}P_1^{\mu\nu}}{64},
\end{equation}
thereby resulting in the spin dipole contribution in \cref{2pm}.

\bibliography{ref.bib}


\end{document}